\documentclass[a4paper,UKenglish]{lipics-v2016}

\usepackage{tikz}
\newcommand{\fixme}[1]{}
\newcommand{\pb}{{\sc $C_3$-Packing-T}}
\newcommand{\perfectpb}{{\sc $C_3$-Perfect-Packing-T}}

\newcommand{\T}{{\cal T}}
\newcommand{\V}{\sigma}
\renewcommand{\S}{S}
\newcommand{\ALG}{\Phi}

\newcommand{\Lft}{Lft}
\newcommand{\Rgt}{Rgt}
\renewcommand{\O}{\mathcal{O}}

\newcommand{\Fb}{{\cal I}}
\newcommand{\Lb}{L}
\newcommand{\Kb}{K}

\newcommand{\Mb}{M}
\newcommand{\Lt}{\tilde{L}}
\newcommand{\Kt}{\tilde{K}}
\newcommand{\Mt}{\tilde{M}}

\newcommand{\m}{\omega} 
\newcommand{\g}{\gamma} 
\newcommand{\gp}{\g'} 
\newcommand{\Pg}{P_{\m,\g}} 
\newcommand{\Pgc}[1]{P^{#1}_{\m,\g}} 
\newcommand{\Xg}[1]{ %les groupes X^1, X^2,... (les sommets blancs) /!\ necessite le package ifthen
\ifthenelse{\equal{#1}{}}{X}{X^{#1}}} 
\newcommand{\Vg}[1]{V^{#1}} 
\newcommand{\Alg}[1]{\alpha^{#1}} 
\newcommand{\mm}[1]{\mu^{#1}}

\newlength{\mylfA}
\newlength{\myhfA}

\newcommand{\fA}[1]{%
\settowidth{\mylfA}{#1}
\def\sizefB{\the\mylfA}
\settoheight{\myhfA}{#1}
\def\heightfB{\the\myhfA}
\hspace{-0.35em}
\begin{tikzpicture}[baseline=(O.base)] 
\draw[ -> ](0,\heightfB)  -- ( \sizefB, \heightfB);
\node (O) at (\sizefB/2,\heightfB/7-0.01ex) {$#1$};
\end{tikzpicture}%
\hspace{-0.35em}
}

\newcommand{\bA}[1]{%
\settowidth{\mylfA}{#1}
\def\sizefB{\the\mylfA}
\settoheight{\myhfA}{#1}
\def\heightfB{\the\myhfA}
\hspace{-0.35em}
\begin{tikzpicture}[baseline=(O.base)] 
\draw[ <- ](0,\heightfB)  -- ( \sizefB, \heightfB);
\node (O) at (\sizefB/2,\heightfB/7-0.01ex) {$#1$};
\end{tikzpicture}%
\hspace{-0.35em}
}

\newtheorem{claim}{Claim}[theorem]
\newtheorem*{claim*}{Claim}
\newtheorem{problem}{Problem}
\newtheorem{observation}[theorem]{Observation}

\makeatletter
\newtheorem*{rep@theorem}{\rep@title}
\newcommand{\newreptheorem}[2]{%
\newenvironment{rep#1}[1]{%
 \def\rep@title{#2 \ref{##1}}%
 \begin{rep@theorem}}%
 {\end{rep@theorem}}}
\makeatother

\newreptheorem{theorem}{Theorem}
\newreptheorem{lemma}{Lemma}
\author[1]{St\'ephane Bessy} 
\author[2]{Marin Bougeret}
\author[3]{Jocelyn Thiebaut} 
\affil[1]{Universit\'e de Montpellier -  CNRS, LIRMM 
\texttt{bessy@lirmm.fr}} 
\affil[2]{Universit\'e de Montpellier - CNRS, LIRMM 
\texttt{bougeret@lirmm.fr}}
\affil[3]{Universit\'e de Montpellier - CNRS, LIRMM 
\texttt{thiebaut@lirmm.fr}} 

\title{Triangle packing in (sparse)
  tournaments: approximation and kernelization.}

\authorrunning{S. Bessy, M. Bougeret and J. Thiebaut} %mandatory. First: Use abbreviated first/middle names. Second (only in severe cases): Use first author plus 'et. al.'

\Copyright{St\'ephane Bessy, Marin Bougeret and Jocelyn Thiebaut}%mandatory, please use full first names. LIPIcs license is "CC-BY";  http://creativecommons.org/licenses/by/3.0/

% mandatory: Please choose ACM 1998 classifications from http://www.acm.org/about/class/ccs98-html . E.g., cite as "F.1.1 Models of Computation". 
% mandatory: Please provide 1-5 keywords

\begin{document}

\maketitle

\begin{abstract}

Given a tournament $\T$ and a positive integer $k$, the \pb~problem
asks if there exists a least $k$ (vertex-)disjoint directed 3-cycles
in $\T$. This is the dual problem in tournaments of the classical
minimal feedback vertex set problem. Surprisingly \pb~did not receive
a lot of attention in the literature. We show that it does not admit
a {\sf PTAS} unless {\sf P}={\sf NP}, even if we restrict the considered instances to sparse tournaments, that is tournaments with a feedback arc set (FAS) being a matching.
Focusing on sparse tournaments we provide a $(1+\frac{6}{c-1})$
approximation algorithm for sparse tournaments having a linear
representation where all the backward arcs have ``length'' at least $c$.
Concerning kernelization, we show that \pb~admits a kernel with
$\O(m)$ vertices, where $m$ is the size of a given feedback arc
set. In particular, we derive a $\O(k)$ vertices kernel for \pb~when
restricted to sparse instances. On the negative size, we show that
\pb~does not admit a kernel of (total bit) size $\O(k^{2-\epsilon})$
unless ${\sf NP} \subseteq {\sf coNP / Poly}$.  The existence of a kernel in
$\O(k)$ vertices for \pb~remains an open question.
\end{abstract}

\section{Introduction and related work}
\subsection*{Tournament}
A tournament $\T$ on $n$ vertices is an orientation of the edges of
the complete undirected graph $K_n$.  Thus, given a tournament
$\T=(V,A)$, where $V = \{v_i, i\in [n]\}$, for each $i,j \in [n]$,
either $v_iv_j \in A$ or $v_jv_i \in A$.  A tournament $\T$ can
alternatively be defined by an ordering $\V(\T)=(v_1,\dots,v_n)$ of
its vertices and a set of \emph{backward arcs} $\bA{A}_{\V}(\T)$
(which will be denoted $\bA{A}(\T)$ as the considered ordering is not
ambiguous), where each arc $a \in \bA{A}(\T)$ is of the form
$v_{i_1}v_{i_2}$ with $i_2 < i_1$.  Indeed, given $\V(\T)$ and
$\bA{A}(\T)$, we can define $V = \{v_i, i\in [n]\}$ and $A= \bA{A}(\T)
\cup \fA{A}(\T)$ where $\fA{A}(\T) = \{v_{i_1}v_{i_2} : (i_1 <
i_2) \mbox{ and } v_{i_2}v_{i_1} \notin \bA{A}(\T)\}$ is the set of
forward arcs of $\T$ in the given ordering $\V(\T)$. In the following,
$(\V(\T), \bA{A}(\T))$ is called a \emph{linear representation} of the
tournament $\T$. For a backward arc $e=v_jv_i$ of $\V(\T)$ the
\emph{span value} of $e$ is $j-i-1$. Then $\mathtt{minspan}(\V(\T))$
(resp. $\mathtt{maxspan}(\V(\T))$) is simply the minimum
(resp. maximum) of the span values of the backward arcs of
$\V(\T)$.\\ A set $A'\subseteq A$ of arcs of $\T$ is a \emph{feedback
  arc set} (or \emph{FAS}) of $\T$ if every directed cycle of $\T$
contains at least one arc of $A'$. It is clear that for any linear
representation $(\V(\T), \bA{A}(\T))$ of $\T$ the set $\bA{A}(\T)$ is
a FAS of $\T$. A tournament is \emph{sparse} if it admits a FAS which
is a matching.
We denote by \pb~the problem of packing the maximum number of vertex
disjoint triangles in a given tournament, where a triangle is a
directed 3-cycle. More formally, an input of \pb~is a tournament $\T$,
an output is a set (called a \emph{triangle packing}) $\S=\{t_i, i \in
[|S|]\}$ where each $t_i$ is a triangle and for any $i \neq j$ we have
$V(t_i) \cap V(t_j) = \emptyset$, and the objective is to maximize
$|S|$. We denote by $opt(\T)$ the optimal value of $\T$.  We denote by
\perfectpb~the decision problem associated to \pb~where an input $\T$
is positive iff there is a triangle packing $\S$ such that
$V(\S)=V(\T)$ (which is called a \emph{perfect triangle packing}).

\subsection*{Related work}

We refer the reader to Appendix where we recall the definitions of the
problems mentioned bellow as well as the standard definitions about
parameterized complexity and approximation.  A first natural related
problem is {\sc 3-Set-Packing}
as we can reduce \pb~to {\sc 3-Set-Packing} by creating an hyperedge
for each triangle. 

\paragraph*{Classical complexity / approximation.}
It is known that \pb~is polynomial if the tournament does not contain
the three forbidden sub-tournaments described in~\cite{cai2002min}.
From the point of view of approximability, the best approximation
algorithm is the $\frac{4}{3}+\epsilon$ approximation
of~\cite{cygan2013improved} for {\sc 3-Set-Packing}, implying the same
result for {\sc $K_3$-packing}
and \pb.  Concerning negative results, it is
known~\cite{guruswami1998vertex} that even {\sc $K_3$-packing} is {\sf
  MAX SNP}-hard on graphs with maximum degree four.  We can also
mention the related "dual" problem {\sc FAST} and {\sc FVST} that received a lot of attention
with for example the {\sf NP}-hardness and {\sf PTAS} for {\sc FAS}
in~\cite{charbit2007minimum} and~\cite{kenyon2007rank} respectively,
and the $\frac{7}{3}$ approximation and inapproximability results for
{\sc FVST} in~\cite{73approx}.

\paragraph*{Kernelization.}
We precise that the implicitly considered parameter here is the size
of the solution.  There is a $\O(k^2)$ vertex kernel for {\sc
  $K_3$-packing} in~\cite{moser2009problem}, and even a $\O(k^2)$
vertex kernel for {\sc 3-Set-Packing} in~\cite{abu2009quadratic},
which is obtained by only removing vertices of the ground set. This remark is important as it directly
implies a $\O(k^2)$ vertex kernel for \pb~(see
Section~\ref{sec:kernel}). Let us now turn to negative results. There
is a whole line of research dedicated to finding lower bounds on the
size of polynomial kernels.  The main tool involved in these bounds is
the weak composition introduced in~\cite{hermelin2012weak} (whose
definition is recalled in Appendix). Weak composition allowed several
almost tight lower bounds, with for example the $\O(k^{d-\epsilon})$
for {\sc $d$-Set-Packing} and $\O(k^{d-4-\epsilon})$ for {\sc
  $K_d$-packing} of ~\cite{hermelin2012weak}. These results where
improved in~\cite{dell2014kernelization} to $\O(k^{d-\epsilon})$ for
\textsc{perfect} $d$-\textsc{Set-Packing}, $\O(k^{d-1-\epsilon})$ for {\sc $K_d$-packing}, and
leading to $\O(k^{2-\epsilon})$ for {\sc perfect $K_3$-packing}.  Notice that negative results for
the "perfect" version of problems (where parameter $k=\frac{n}{d}$,
where $d$ is the number of vertices of the packed structure) are
stronger than for the classical version where $k$ is arbitrary.
Kernel lower bound for these "perfect" versions is sometimes referred
as \emph{sparsification lower bounds}.

\subsection*{Our contributions}
Our objective is to study the approximability and kernelization of
\pb.  On the approximation side, a natural question is a possible
improvement of the $\frac{4}{3}+\epsilon$ approximation implied by
{\sc 3-Set-Packing}.  We show that, unlike {\sc FAST}, \pb~does not
admit a {\sf PTAS} unless {\sf P}={\sf NP}, even if the tournament is
sparse.  We point out that, surprisingly, we were not able to find any
reference establishing a negative result for \pb, even for the {\sf
  NP}-hardness.  As these results show that there is not much room for
improving the approximation ratio, we focus on sparse tournaments and
followed a different approach by looking for a condition that would
allow ratio arbitrarily close to $1$. In that spirit, we provide a
$(1+\frac{6}{c-1})$ approximation algorithm for sparse tournaments
having a linear representation with $\mathtt{minspan}$  at least
$c$.
Concerning kernelization, we complete the panorama of sparsification
lower bounds of~\cite{jansen2015sparsification} by proving that
\perfectpb~does not admit a kernel of (total bit) size
$\O(n^{2-\epsilon})$ unless ${\sf NP} \subseteq {\sf coNP / Poly}$.
This implies that \pb~does not admit a kernel of (total bit) size
$\O(k^{2-\epsilon})$ unless ${\sf NP} \subseteq {\sf coNP / Poly}$.
We also prove that \pb~admits a kernel of $\O(m)$ vertices, where $m$
is the size of a given FAS of the instance, and that
\pb~restricted to sparse instances has a kernel in $\O(k)$ vertices
(and so of total size bit $\O(k\log (k))$).  The existence of a kernel
in $\O(k)$ vertices for the general \pb~remains our main open
question.

\section{Specific notations and observations}
\label{sec:notation}

Given a linear representation $(\V(\T),\bA{A}(\T))$ of a tournament
$\T$, a triangle $t$ in $\T$ is a triple $t=(v_{i_1},v_{i_2},v_{i_3})$
with $i_l < i_{l+1}$ such that either $v_{i_3}v_{i_1} \in \bA{A}(\T)$,
$v_{i_3}v_{i_2} \notin \bA{A}(\T)$ and $v_{i_2}v_{i_1} \notin
\bA{A}(\T)$ (in this case we call $t$ a \emph{triangle with backward
  arc} $v_{i_3}v_{i_1}$), or $v_{i_3}v_{i_1} \notin \bA{A}(\T)$,
$v_{i_3}v_{i_2} \in \bA{A}(\T)$ and $v_{i_2}v_{i_1} \in \bA{A}(\T)$
(in this case we call $t$ a \emph{triangle with two backward arcs}
$v_{i_3}v_{i_2}$ and $v_{i_2}v_{i_1}$).

Given two tournaments $\T_1, \T_2$ defined by $\V(\T_l)$ and
$\bA{A}(\T_l)$ we denote by $\T=\T_1\T_2$ the tournament called the
concatenation of $\T_1$ and $\T_2$, where $\V(\T) = \V(\T_1)\V(\T_2)$
is the concatenation of the two sequences, and $\bA{A}(\T) =
\bA{A}(\T_1) \cup \bA{A}(\T_2)$. Given a tournament $\T$ and a subset
of vertices $X$, we denote by $\T \setminus X$ the tournament
$\T[V(\T) \setminus X]$ induced by vertices $V(\T) \setminus X$, and
we call this operation \emph{removing $X$ from $\T$}.  Given an arc
$a=uv$ we define $h(a)=v$ as the head of $a$ and $t(a)=u$ as the tail
of $a$.  Given a linear representation $(V(\T),\bA{A}(\T))$ and an arc
$a \in \bA{A}(\T)$, we define $s(a) = \{v : h(a) < v < t(a)\}$ as the
\emph{span} of $a$.  Notice that the span value of $a$ is then exactly
$|s(a)|$.  \\ Given a linear representation $(V(\T),\bA{A}(\T))$ and a
vertex $v \in V(\T)$, we define the degree of $v$ by $d(v)=(a,b)$,
where $a = |\{vu \in \bA{A}(\T) : u < v\}|$ is called the \emph{left
  degree} of $v$ and $b = |\{uv \in \bA{A}(\T) : u > v\}|$ is called
the \emph{right degree} of $v$. We also define \mbox{$V_{(a,b)} = \{v
  \in V(\T)| d(v)=(a,b)\}$}.  Given a set of pairwise distinct pairs
$D$, we denote by \pb$^D$ the problem \pb~restricted to tournaments
such that there exists a linear representation where $d(v) \in D$ for
all $v$.  Notice that when $D_{M}=\{(0,1),(1,0),(0,0)\}$, instances of
\pb$^{D_M}$ are the sparse tournaments.\\ Finally let us point out
that it is easy to decide in polynomial time if a tournament is sparse or not, and if so,
to give a linear representation whose FAS is a
matching. The corresponding algorithm is detailed in Appendix in Lemma~\ref{lem:faslinear}. 
Thus, in the
following, when considering a sparse tournament we will assume that a
linear ordering of it where backward arcs form a matching is also
given.
\section{Approximation for sparse tournaments}
\label{sec:approx}

\subsection{{\sf APX}-hardness for sparse tournaments}

In this subsection we prove that \pb$^{D_M}$ is {\sf APX}-hard by providing a
$L$-reduction (see Definition~\ref{def:L} in appendix) from Max
2-SAT(3), which is known to be {\sf APX}-hard~\cite{ausiello2012complexity,berman1999some}.  Recall that in
the {\sc Max 2-SAT(3)} problem where each clause contains exactly $2$
variables and each variable appears in at most 3 clauses (and at
most twice positively and once negatively).

\paragraph*{Definition of the reduction} 
\label{subsec:reduction2}
Let $\cal{F}$ be an instance of {\sc Max 2-SAT(3)}. In the following,
we will denote by $n$ the number of variables in $\cal{F}$ and $m$ the
number of clauses. Let $\{x_i, 1 \in [n]\}$ be the set of variables of
$\cal{F}$ and $\{C_j, j \in [m]\}$ its set of clauses.

We now define a reduction $f$ which maps an instance ${\cal F}$ of {\sc Max 2-SAT(3)} to an instance ${\cal T}$ of \pb$^{D_M}$. 
For each variable $x_i$ with $i \in [n]$, we create a tournament $L_i$ as follows and we call it \emph{variable gadget}. We refer the reader to Figure~\ref{fig:li} where an example of variable gadget is depicted. 
Let $\V(L_i) = (X_i, X'_i, \overline{X_i}, \overline{X_i}', \{\beta_i\}, \{\beta'_i\}$ $, A_i, B_i, \{\alpha_i\}, A'_i, B'_i)$.
We define \mbox{$C=\{ X_i,X'_i,\overline{X_i},\overline{X_i}',A_i,B_i,A'_i,B'_i\}$}. All sets of $C$ have size $4$. We denote $X_i = (x_i^1,x_i^2,x_i^3,x_i^4)$, and we extend the notation in a straightforward manner to the other others sets of $C$. 
Let us now define $\bA{A}(L_i)$. For each set of $C$, we add a backward arc whose head is the first element and the tail is the last element (for example for $X_i$ we add the arc $x_i^4x_i^1$). 
Then, we add to $\bA{A}(L_i)$ the set $\{e_1,e_2,e_3,e_4\}$ where $e_1=x_i^3a_i^3$, $e_2 = x_i^{'3}a_i^{'3}$, $e_3 = \overline{x_i^3} b_i^3$, $e_4 = \overline{x_i^{'3}} b_i^{'3}$ 
and the set $\{m_1,m_2\}$ where $m_1 = a_i^{'2}a_i^2$, $m_2 = b_i^{'2}b_i^2$ called the two \emph{medium arcs} of the variable gadget.
This completes the description of tournament $L_i$. Let $L = L_1 \dots L_n$ be the concatenation of the $L_i$. 

\begin{figure}[!htbp]
\begin{center}
\includegraphics[width=0.75\textwidth]{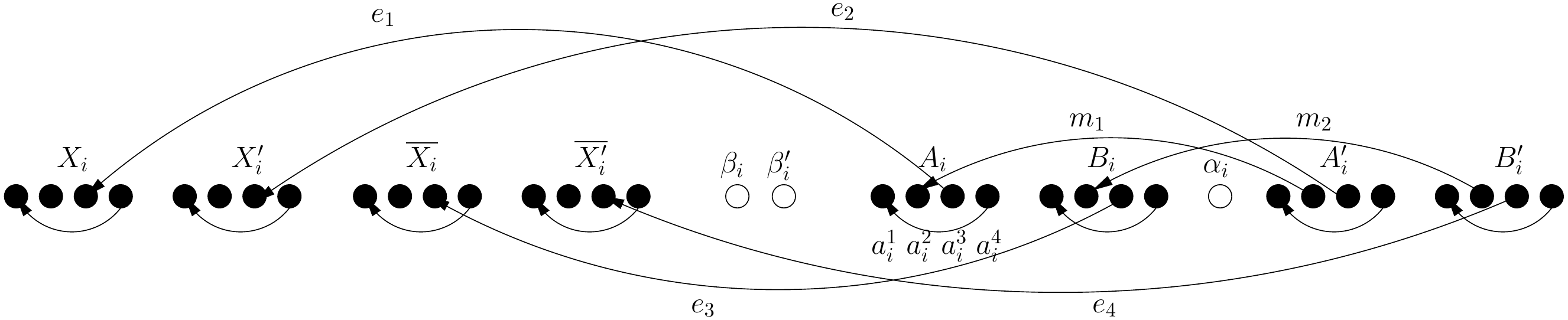}
\end{center}
\caption{Example of a variable gadget $L_i$.}
\label{fig:li}
\end{figure}

For each clause $C_j$ with $j \in [1,m]$, we create a tournament $K_j$ with ordering $\V(K_i) = (\theta_j, d^1_j,c^1_j,c^2_j,d^2_j)$ and $\bA{A}(K_i) = \{d^2_jd^1_j\}$.
We also define $K = K_1\dots K_m$. 
Let us now define $\T = LK$. We add to $\bA{A}(\T)$ the following backward arcs from $V(K)$ to $V(L)$. If $C_j = l_{i_1} \vee l_{i_2} $ is a clause in $\cal{F}$ then we add the arcs $c_j^1v_{i_1}, c_j^2v_{i_2}$ where $v_{i_c}$ is the vertex in $\{x_{i_c}^2,x_{i_c}^{'2},\overline{x_{i_c}^2}\}$ corresponding to $l_{i_c}$: if $l_{i_c}$ is a positive occurrence of variable $i_c$ we chose
$v_{i_c} \in \{x_{i_c}^2,x_{i_c}^{'2}\}$, otherwise we chose $v_{i_c} = \overline{x_{i_c}^2}$. Moreover, we chose vertices $v_{i_c}$ in such a way that for any $i \in [n]$, for each $v \in \{x_i^2,x_i^{'2},\overline{x_i^2}\}$ there exists a unique arc $a \in \bA{A}(\T)$ such that $h(a)=v$. This is always possible as each variable has at most two positive occurrences and one negative occurrence.
Thus, $x_i^2$ represent the first positive occurrence of variable $i$, and $x_i^{'2}$ the second one. We refer the reader to Figure~\ref{fig:LetK} where an example of the connection between variable and clause gadget is depicted.
\begin{figure}[!htbp]
\begin{center}
\includegraphics[width=0.75\textwidth]{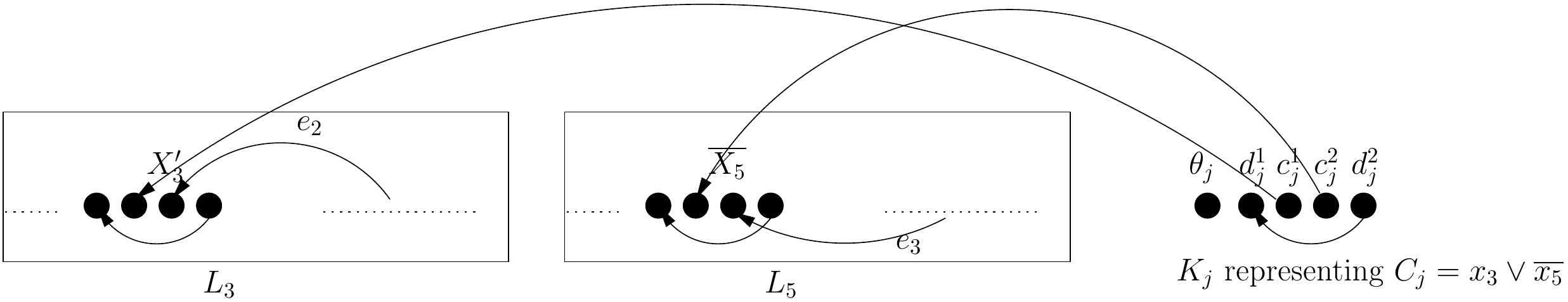}
\end{center}
\caption{Example showing how a clause gadget is attached to variable gadgets.}
\label{fig:LetK}
\end{figure}

Notice that vertices of $\overline{X'_i}$ are never linked to the clauses gadget. However, we need this set to keep the variable gadget symmetric so that setting $x_i$ to true or false leads to the same number of triangles inside $L_i$. This completes the description of $\T$. Notice that the degree of any vertex is in $\{(0,1),(1,0),(0,0)\}$, and thus $\T$ is an instance of \pb$^{D_M}$.

Let us now distinguish three different types of triangles in $\T$. A triangle $t=(v_1,v_2,v_3)$ of $\T$ is called an \emph{outer} triangle iff $\exists j \in [m]$ such that 
$v_2 = \theta_j$ and $v_3 = c^l_j$ (implying that $v_1 \in V(L)$), \emph{variable inner} iff $\exists i \in [n]$ such that $V(t) \subseteq V(L_i)$, 
and \emph{clause inner} iff $\exists j \in [m]$ such that $V(t) \subseteq V(K_j)$.
Notice that a triangle $t=(v_1,v_2,v_3)$ of $\T$ which is neither outer, variable or clause inner has necessarily $v_3 = c^l_j$ for some $j$, and $v_2 \neq \theta_j$ ($v_2$ could be in $V(L)$ or $V(K)$).
In the following definition, for any $Y \in C$ (where $C=\{ X_i,X'_i,\overline{X_i},\overline{X_i}',A_i,B_i,A'_i,B'_i\}$) with  $Y=(y^1,y^2,y^3,y^4)$, we define $t_Y^2 = (y^1,y^2,y^4)$ and $t_Y^3 = (y^1,y^3,y^4)$. For example, $t_{X'_i}^2 = (x_i^{'1},x_i^{'2},x_i^{'4})$.
For any $i\in [n]$, we define $P_i$ and $\overline{P_i}$, two sets of vertex disjoint variable inner triangles of $V(L_i)$, by:
\begin{itemize}
	\item $P_i = \{t_{X_i}^3, t_{X'_i}^3, t_{\overline{X_i}}^2, t_{\overline{X'_i}}^2, t_{A_i}^3,  t_{B_i}^2,  t_{A'_i}^3,  t_{B'_i}^2, (h(e_3),\beta_i,t(e_3)), (h(e_4),\beta'_i,t(e_4)), (h(m_1),\alpha_i,t(m_1))\}$
	\item $\overline{P_i} = \{t_{X_i}^2, t_{X'_i}^2, t_{\overline{X_i}}^3, t_{\overline{X'_i}}^3, t_{A_i}^2,  t_{B_i}^3,  t_{A'_i}^2,  t_{B'_i}^3, (h(e_1),\beta_i,t(e_1)), (h(e_2),\beta'_i,t(e_2)), (h(m_2),\alpha_i,t(m_2))\}$
\end{itemize}
Notice that $P_i$ (resp. $\overline{P_i}$) uses all vertices of $L_i$ except $\{x_i^2,x_i^{'2}\}$ (resp. $\{\overline{x_i^2},\overline{x_i^{'2}}\}$).
For any $j \in [m]$, and $x \in [2]$ we define the set of clause inner triangle of $K_j$, that is $Q^x_j = \{(d^1_j,c^x_j,d^2_j)\}$.

Informally, setting variable $x_i$ to true corresponds to create the $11$ triangles of $P_i$ in $L_i$ (as leaving vertices $\{x^2_i,x^{2'}_i\}$ available allows to create outer triangles corresponding to satisfied clauses), and setting it to false corresponds to create the $11$ triangles of $\overline{P_i}$. Satisfying a clause $j$ using its $x^{th}$ literal (represented by a vertex $v \in V(L)$) corresponds to create triangle in $Q^{3-x}_j$ as it leaves $c^x_j$ available to create the triangle $(v,\theta_j,c^x_j)$. Our final objective (in Lemma~\ref{lem:Lreducv2}) is to prove that satisfying $k$ clauses is equivalent to find $11n+m+k$ vertex disjoint triangles.

\paragraph*{Restructuration lemmas}

Given a solution $\S$, let $I^{L}_i =\{t \in \S : V(t) \subseteq V(L_i)\}$, $I^{K}_j =\{t \in \S : V(t) \subseteq V(K_j)\}$, 
$I^{L} = \cup_{i \in [n]} I^L_i$ be the set of variable inner triangles of $\S$, $I^{K} = \cup_{j \in [m]} I^K_j$ be the set of clause inner triangles of $\S$, and 
$O = \{t \in \S \mbox{ $t$ is an outer triangle }\}$ be the set of outer triangles of $\S$. Notice that \textit{a priori} $I^L,I^K,O$ does not necessarily form a partition of $\S$.
However, we will show in the next lemmas how to restructure $\S$ such that $I^L,I^K,O$ becomes a partition.

 \begin{lemma}\label{lem:intv2}
For any $\S$ we can compute in polynomial time a solution $\S' =
\{t'_l, l\in [k]\}$ such that $|\S'| \ge |\S|$ and for all $j\in[m]$
there exists $x \in [2]$ such that $I^{'K}_j=Q^x_j$ and
\begin{itemize}
\item either $\S'$ does not use any other vertex of $K_j$ ($V(\S') \cap V(K_j) = V(Q^x_j)$)
\item either $\S'$ contains an outer triangle $t'_l=(v,\theta_j, c^{3-x}_j)$ with $v \in V(L)$ (implying $V(\S') \cap V(K_j) = V(K_j)$)
\end{itemize}
 \end{lemma}

 \begin{proof}
 Consider a solution $\S = \{t_l,l \in [k]\}$.
Let us suppose that $\S$ does not verify the desired property.
We say that $j \in [m]$ satisfies $(\star)$ iff there exists $x \in [2]$ such that $I^{K}_j=Q^x_j$ and
either $\S$ does not use any other vertex of $K_j$, or $\S$ contains an outer triangle $t_l=(v,\theta_j, c^{3-x}_j)$ with $v \in V(L)$.

 Let us restructure $\S$ to increase the number of $j$ satisfying
 $(\star)$, which will be sufficient to prove the lemma.
 Consider the largest $j\in [m]$ which does not satisfy $(\star)$.
Let $c = |I^{K}_j|$. Notice that the only possible triangle of $I^{K}_j$ contains $a=d^2_jd^1_j$, implying $c \le 1$.

If $c=1$, let $t \in I^K_j$ and $v_0 = \{c^1_j,c^2_j\} \setminus V(t)$.
If $v_0 \notin V(\S)$, then let us prove that $\theta_j \notin V(\S)$. Indeed, by contradiction if $\theta_j \in V(S)$, let $t' \in \S$ such that $\theta_j \in V(t')$. As $d(\theta_j)=(0,0)$ we necessarily have
$t'=(u,\theta_j,w)$ with $w = c^{x'}_{j'}$ with $j' \ge j$, which contradicts the maximality of $j$.
Otherwise ($v_0 \in V(\S)$), then denoting by $t'$ the triangle of $\S$ which contains $v_0$ we must have $t'=(u,v,v_0)$.
 Indeed, we cannot have (for some $u', v'$) $t'=(v_0,u',v')$ as there is no backward arc $a$ with $h(a)=v_0$ and we cannot have either $t'=(u',v_0,v')$ as this would imply $v'=c^{x'}_{j'}$ for $j' > j$ and again contradict the definition of $j$. As, again, by maximality of $j$ we get $\theta_j \notin V(\S)$ (and since $u\theta_j$ and $\theta_jv_0$ are forward arcs), we can replace $t'$ by the triangle $(u,\theta_j,v_0)$ which is disjoint to the other triangles of $\S$.

If $c=0$. Notice first that by maximality of $j$, $d^2_j \notin V(\S)$ as $d^2_j$ could only be used in a triangle $t=(v,d^2_j,c^x_{j'})$ with $j' > j$.
Let $Z = V(\S) \cap \{c^1_j,c^2_j\}$. 
If $|Z|=0$, then by maximality of $j$ we get $d^1_j \notin V(\S)$ and $\theta_j \notin V(\S)$, and thus we add to $\S$ triangle $(d^1_j,c^1_j,d^2_j)$.
If $|Z|=1$, let $c^x_j \in Z$ and $t \in \S$ such that $c^x_j \in V(t)$. By maximality of $j$ we necessarily have $t=(u,v,c^x_j)$ for some $u,v$.
If $v \neq \theta_j$ then by maximality of $j$ we have $\theta_j \notin V(\S)$, and thus we swap $v$ and $\theta_j$ in $t$ and now suppose that $\theta_j \in V(t)$. This implies that $d^1_j \notin V(\S)$
(before the swap we could have had $v = d^1_j$, but now by maximality of $j$ we know that $d^1_j$ is unused), and we add $(d^1_j,c^{3-x}_j,d^2_j)$ to $\S$.
It only remains now case where $|Z|=2$. If there exists $t \in \S$ with $Z \subseteq V(t)$, then $t=(u,c^1_j,c^2_j)$. Using the same arguments as above we get that $\{\theta_j,d^1_j\} \cap V(\S) = \emptyset$, 
and thus we swap $c^1_j$ by $\theta_j$ in $t$ and add $(d^1_,c^1_j,d^2_j)$ to $\S$.
Otherwise, let $t_x \in \S$ such that $c^x_j \in V(t_x)$ for $x \in [2]$. This implies that $t_x = (u_x,v_x,c^x_j)$. If $\theta_j \notin V(t_1) \cup V(t_2)$ then $\theta_j \notin V(\S)$ and we swap $v_1$ with $\theta_j$. Therefore, from 
now on we can suppose that $\theta_j \in V(t_x)$ for $x \in [2]$. Then, if $d^1_j \notin V(t_{3-x})$ then $d^1_j \notin V(\S)$ and thus we swap $v_{3-x}$ with $d^1_j$ and we now assume that
  $d^1_j \in V(t_{3-x})$. Finally, we remove $t_{3-x}$ from $\S$ and add instead $(d^1_j,c^{3-x}_j,d^2_j)$.
 \end{proof}

\begin{corollary}\label{cor:outerinnerv2}
For any $\S$ we can compute in polynomial time a solution $\S'$ such that $|\S'| \ge |\S|$, and $\S'$
 only contains outer, variable inner, and clause inner triangles. Indeed, in the solution $\S'$ of Lemma~\ref{lem:intv2}, given any $t \in \S'$, either
$V(t)$ intersects $V(K_j)$ for some $j$ and then $t$ is an outer or a clause inner triangle, or $V(t) \subseteq V(L_i)$ for $i \in [n]$ as there is no backward arc $uv$ with $u \in V(L_{i_1})$ and $v \in V(L_{i_2})$ with $i_1 \neq i_2$ .
\end{corollary}

\begin{lemma}
\label{lem:goodpatternv2}
For any $\S$ we can compute in polynomial time a solution $\S'$ such
that $|\S'| \ge |\S|$, $\S'$ satisfies Lemma~\ref{lem:intv2}, and for
every $i \in[n]$, $I^{'L}_i = P_i$ or $I^{'L}_i = \overline{P_i}$.
\end{lemma}

\begin{proof}
Let $\S_0$ be an arbitrary solution, and $\S$ be the solution obtained
from $\S_0$ after applying Lemma~\ref{lem:intv2}.
By Corollary~\ref{cor:outerinnerv2}, we partition $\S$ into $\S = I^L \cup I^K \cup O$.
Let us say that $i \in [n]$ satisfies $(\star)$ if $I^L_i = P_i$ or $I^L_i = \overline{P_i}$.
Let us suppose that $\S$ does not verify the desired property, and show how to restructure $\S$ to increase the number of $i$ satisfying $(\star)$ while still satisfying Lemma~\ref{lem:intv2}, which will prove the lemma.

Let $\Lft_i = X_i \cup X'_i \cup \overline{X_i} \cup \overline{X'_i}$ and $\Rgt_i = A_i \cup B_i \cup \{\alpha_i\} \cup  A'_i \cup B'_i$ be two subset of vertices of $V(L_i)$.
Given any solution $\tilde{\S}$ satisfying Lemma~\ref{lem:intv2}, we define the following sets. 
Let $\tilde{\S}^{\Lft_i} = \{t \in \tilde{I}^L_i : V(t) \subseteq \Lft_i \}$, $\tilde{\S}^{\Rgt_i} = \{t \in \tilde{I}^L_i : V(t) \subseteq \Rgt_i \}$, and  
\mbox{$\tilde{\S}^{\Lft_i\Rgt_i} = \{t \in \tilde{I}^L_i : V(t) \cap \Lft_i \neq \emptyset \mbox{ and } V(t) \cap \Rgt_i \neq \emptyset\}$}. Observe that these three sets define a partition of $\tilde{I}^L_i$, 
and that triangles of $\tilde{\S}^{\Lft_i}$ are even included in $W$ with $W \in \{X_i,  X'_i , \overline{X_i}, \overline{X_i}'\}$.
Let $\tilde{\S}^{O_i} = \{t \in \tilde{O} : V(t) \cap V(L_i) \neq \emptyset\}$ be the set of outer triangles of $\tilde{\S}$ intersecting $L_i$.
We also define $g_i(\tilde{\S})=(|\tilde{\S}^{\Lft_i}|,|\tilde{\S}^{\Lft_i\Rgt_i}|,|\tilde{\S}^{\Rgt_i}|,|\tilde{\S}^{O_i}|)$ and 
$h_i(\tilde{S})=|\tilde{\S}^{\Lft_i}|+|\tilde{\S}^{\Lft_i\Rgt_i}|+|\tilde{\S}^{\Rgt_i}|+|\tilde{\S}^{O_i}|=|\tilde{I}^L_i \cup \tilde{\S}^{O_i}|$.

Our objective is to restructure $\S$ into a solution $\S'$ with $\S' = (\S \setminus (I^L_i \cup \S^{O_i})) \cup (I^{'L}_i \cup \S^{'O_i})$. We will 
define $I^{'L}_i$ and $\S^{'O_i}$ verifying the following properties $(\triangle)$:
\begin{description}
\item[$\triangle_1:$] $I^{'L}_i = P_i$ or $I^{'L}_i=\overline{P_i}$,
\item[$\triangle_2:$] $\S^{'O_i} \subseteq \S^{O_i}$, 
\item[$\triangle_3:$] $|(I^{'L}_i \cup \S^{'O_i})| \ge |(I^L_i \cup \S^{O_i})| $ (which is equivalent to $h_i(\S') \ge h_i(\S) $), 
\item[$\triangle_4:$] triangles of $I^{'L}_i \cup \S^{'O_i}$ are vertex disjoint. 
\end{description}
 Notice that $\triangle_2$ and $\triangle_4$ imply that all triangles of $\S'$ are still vertex disjoint. Indeed, as $\S$ satisfies Lemma~\ref{lem:intv2}, the only triangles of $\S$ intersecting 
 $L_i$ are $I^L_i \cup \S^{O_i}$, and thus replacing them with $I^{'L}_i \cup \S^{'O_i}$ satisfying the above property implies that all triangles of $\S'$ are vertex disjoint. Moreover, $\S'$ will still satisfy Lemma~\ref{lem:intv2} even with $\S^{'O_i} \subseteq \S^{O_i}$ as removing outer triangles cannot violate property of Lemma~\ref{lem:intv2}. 
Finally $\triangle_3$ implies that $|\S'| \ge |\S|$. Thus, defining  $I^{'L}_i$ and $\S^{'O_i}$ satisfying $(\triangle)$ will be sufficient to prove the lemma. Let us now state some useful properties.

\begin{description}
\item[$p_1:$] $|\S^{\Lft_i}| \le 4$
\item[$p_2:$] $|\S^{\Lft_i\Rgt_i}| \le 4$ as for any $t \in \S^{\Lft_i\Rgt_i}$ there exists $l \in [4]$ such that $V(t) \supseteq V(e_l)$. 
\item[$p_3:$] $|\S^{\Rgt_i}| \le 5$ (as $|V(\S^{\Rgt_i})| = 17$). 
Let $Z = V(\S^{\Lft_i\Rgt_i}) \cap \Rgt_i$.
Let us also prove that if $Z \supseteq \{a^3_i,b^3_i\}$, $Z \supseteq \{a^{'3}_i,b^{'3}_i\}$, $Z \supseteq \{a^3_i,b^{'3}_i\}$ or $Z \supseteq \{a^{'3}_i,b^3_i\}$ then $|\S^{\Rgt_i}| \le 4$. 
For any $W \in \{A_i, B_i,  A'_i, B'_i\}$, let $s_W$ be the unique arc $a$ of $\T$ such that $V(a) \subseteq W$ and let $m_W$ be the unique medium arc $a$ such that $V(a) \cap W \neq \emptyset$. 
Let us call the $\{s_W\}$ the four small arcs of the tournament induced by $\Rgt_i$. 
%or example, if $W=A_i$ then $s_W = (a_i^4a_i^1)$).
Let $\bA{A}(\S^{\Rgt_i}) = \{a \in \bA{A}(L_i) : \exists t \in \S^{\Rgt_i}$ \mbox{such that } $V(a) \subseteq V(t) \}$ be the set of backward arcs used by $\S^{\Rgt_i}$.
Observe that arcs of $\bA{A}(\S^{\Rgt_i})$ are small or medium arcs. Let us bound $|\bA{A}(\S^{\Rgt_i})|=|\S^{\Rgt_i}|$.
Notice that for any $W \in \{A_i, B_i,  A'_i, B'_i\},$ $W \cap Z \neq \emptyset$ implies that $\bA{A}(\S^{\Rgt_i})$ cannot contain both $s_W$ and $m_W$.
If $\S^{\Rgt_i}$ contains the $4$ small arcs then by previous remark $\S^{\Rgt_i}$ cannot contain any medium arc, 
and thus $|\S^{\Rgt_i}| \le 4$. If $\S^{\Rgt_i}$ contains $3$ small arcs then it can only contain one medium arc, implying  $|\S^{\Rgt_i}| \le 4$. Obviously, if $|\S^{\Rgt_i}|$ contains $2$ or less small arcs then $|\S^{\Rgt_i}| \le 4$. 
\item[$p_4:$] property $p_3$ implies that if $|\S^{\Lft_i\Rgt_i}| \ge 3$, or if $|\S^{\Lft_i\Rgt_i}|=2$ and triangles of $\S^{\Lft_i\Rgt_i}$ contain $\{e_1,e_3\}$, $\{e_1,e_4\}$, $\{e_2,e_3\}$ or $\{e_2,e_4\}$, 
then $|\S^{\Rgt_i}| \le 4$ (where triangles of $\S^{\Lft_i\Rgt_i}$ contains $\{e_i,e_j\}$ means that there exist $t_1,t_2$ in $\S^{\Lft_i\Rgt_i}$ such that $V(t_1) \supseteq V(e_i)$ and $V(t_2) \supseteq V(e_j)$).
\item[$p_5:$] $|\S^{O_i}| \le 3$. Moreover, if $|\S^{\Lft_i}|=4$ then $|\S^{O_i}| \le 4 - |\S^{\Lft_i\Rgt_i}|$, and if $|\S^{\Lft_i}|=3$ and $|\S^{\Lft_i\Rgt_i}|=4$ then $|\S^{O_i}| \le 1$. The last two inequalities come from the fact that for any $W \in \{X_i, X'_i, \overline{X_i}, \overline{X'_i}\}$, we cannot have both $t_1 \in \S^{O_i}$, $t_2 \in \S^{\Lft_i\Rgt_i}$ 
and $t_3 \in \S^{\Lft_i}$ with $V(t_i) \cap W \neq \emptyset$. 
\end{description}

Notice that if a solution $\S'$ satisfies $I^{'L}_i = P_i$ or $I^{'L}_i=\overline{P_i}$ then $g_i(\S')=(4,2,5,z)$ where $z \in [2]$, and $h_i(\S')=11+z$.
In the following we write $(u^1_1,u^1_2,u^1_3,u^1_4) \le (u^2_1,u^2_2,u^2_3,u^2_4)$ iff $u^1_i \le u^2_i$ for any $i \in [4]$.
Let us describe informally the following argument which will be used several times. Let $z=|\S^{O_i}|$. If $z \le 1$ or if $z = 2$ but the two corresponding outer triangles do not use one vertex in $X_i \cup X'_i$ and one vertex in $\overline{X_i}$, then we will able to "save" all these outer triangles (while creating the optimal number of variable inner triangles in $L_i$), meaning that $\S^{'O_i} = \S^{O_i}$, as either $P_i$ or $\overline{P_i}$ will leave vertices of $\S^{O_i} \cap \Lft_i$ available for outer triangles. Let us proceed by case analysis according to the value $|\S^{\Lft_i\Rgt_i}|$. Remember that $|\S^{\Lft_i\Rgt_i}| \le 4$ according to $p_2$.

Case 1: $|\S^{\Lft_i\Rgt_i}| \le 1$. According to  $p_1, p_3$ we get $g_i(\S) \le (4,1,5,z)$ where $z \in [3]$.
In this case, $\S^{'O_i} = \S^{O_i} \setminus \{t \in \S : V(t) \ni \overline{x^2_i}\}$ and $I^{'L}_i = P_i$ verify $(\triangle)$.
In particular, we have $h_i(\S') \ge h_i(\S)$ as $g_i(\S') \ge (4,2,5,z-1)$.

Case 2: $|\S^{\Lft_i\Rgt_i}|=2$. Let $g_i(\S) = (x,2,y,z)$. If $x \le 3$, then $g_i(\S) \le (3,2,5,z)$ by $p_3$ and we set
$\S^{'O_i} = \S^{O_i} \setminus \{t \in \S : V(t) \ni \overline{x^2_i}\}$ and $I^{'L}_i = P_i$. This satisfies $(\triangle)$
as in particular we have $h_i(\S') \ge h_i(\S)$ as $g_i(\S') \ge (4,2,5,z-1)$.
Let us now turn to case where $x=4$. Let $\S^{\Lft_i\Rgt_i}=\{t_1,t_2\}$. Let us first suppose that triangles of $\S^{\Lft_i\Rgt_i}$ contain $\{e_i,e_j\}$ with 
$\{e_i,e_j\} \in \{\{e_1,e_3\},\{e_1,e_4\},\{e_2,e_3\},\{e_2,e_4\}\}$. 
By $p_4$ we get  $y \le 4$, implying $g_i(\S) \leq (4,2,4,z)$.
In this case, $\S^{'O_i} = \S^{O_i} \setminus \{t \in \S : V(t) \ni \overline{x^2_i}\}$ and $I^{'L}_i = P_i$ verify $(\triangle)$.
In particular, we have $h_i(\S') \ge h_i(\S)$ as $g_i(\S')=(4,2,5,z-1)$. Let us suppose now that $t_1$ contains $e_1$ and $t_2$ contains $e_2$ (case (2a)), or 
$t_1$ contains $e_3$ and $t_2$ contains $e_4$ (case (2b)). In both cases we have $g_i(\S) \le (4,2,5,z)$ where $z \in [2]$ by $p_5$. 
More precisely, $p_5$ implies that $\{W \in \{X_i, X'_i, \overline{X_i}, \overline{X'_i}\} : W \cap V(\S^{O_i})\} \neq \emptyset$ is included in $\{X_,X'_i\}$ (case 2b) or in $\overline{X_i}$ (case 2a).
Thus, in case (2a) we define  $\S^{'O_i} = \S^{O_i}$ and $I^{'L}_i = \overline{P_i}$. In case (2b) we define  $\S^{'O_i} = \S^{O_i}$ and $I^{'L}_i = P_i$.
In both cases these sets verify $(\triangle)$ as in particular $g_i(\S') = (4,2,5,z)$.

\begin{sloppypar}
Case 3: $|\S^{\Lft_i\Rgt_i}|=3$. In this case $g_i(\S) \le (x,3,4,z)$ by $p_4$.
	If $x \le 3$, the sets \mbox{$\S^{'O_i} = \S^{O_i} \setminus \{t \in \S : V(t) \ni \overline{x^2_i}\}$} and $I^{'L}_i = P_i$ verify $(\triangle)$.
In particular, we have \mbox{$h_i(\S') \ge h_i(\S)$} as $g_i(\S') \ge (4,2,5,z-1)$.
If $x = 4$ then $z \le 1$ by $p_5$. Thus, we define $I^{'L}_i = P_i$ if $V(\S^{O_i}) \cap (X_i \cup X'_i) \neq \emptyset$, and 
$I^{'L}_i = \overline{P_i}$ otherwise, and $\S^{'O_i} = \S^{O_i}$. These sets satisfy $(\triangle)$ as in particular $g_i(\S') = (4,2,5,z)$.
\end{sloppypar}

Case 4: $|\S^{\Lft_i\Rgt_i}|=4$. Let $g_i(\S) = (x,4,y,z)$.
If $x=4$ then $z \le 0$ by $p_5$ and $y \le 3$ as $x+4+y \le \frac{|V(L_i)|}{3}$.

Thus, we set $\S^{'O_i} = \S^{O_i} = \emptyset$, $I^{'L}_i = P_i$ (which is arbitrary in this case), and we have property $(\triangle)$ as  $g_i(\S') \ge (4,2,5,0)$.
If $x=3$ (this case is depicted Figure~\ref{fig:ex_restruct_Li}) then $y \le 4$ by $p_3$ and $z \le 1$ by $p_5$, implying $g_i(\S) = (3,4,4,z)$. Thus, we define $I^{'L}_i = P_i$ if $V(\S^{O_i}) \cap (X_i \cup X'_i) \neq \emptyset$, and 
$I^{'L}_i = \overline{P_i}$ otherwise, and $\S^{'O_i} = \S^{O_i}$. These sets satisfy $(\triangle)$ as in particular $g_i(\S') = (4,2,5,z)$.
Finally, if $x \le 2$ then $g_i(\S) \le (2,4,4,z)$ by $p_3$. In this case, $\S^{'O_i} = \S^{O_i} \setminus \{t \in \S : V(t) \ni \overline{x^2_i}\}$ and $I^{'L}_i = P_i$ verify $(\triangle)$.
In particular, we have $h_i(\S') \ge h_i(\S)$ as $g_i(\S') \ge (4,2,5,z-1)$. 

\begin{figure}[!htbp]
\begin{center}
\includegraphics[width=0.80\textwidth]{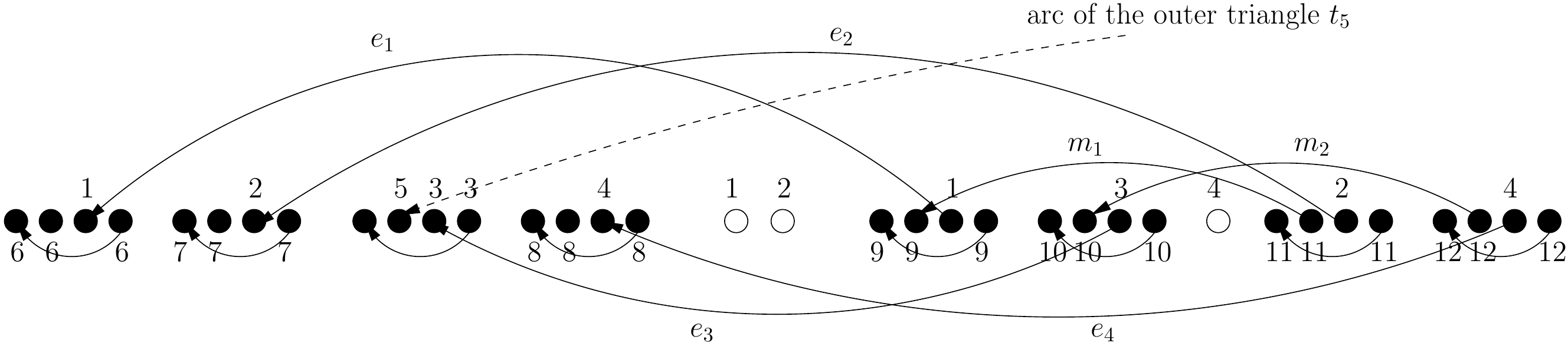}
\end{center}
\caption{Example showing a "bad shaped" solution of case $4$ with  $g_i(\S) = (3,4,4,1)$. We have $\S^{\Lft_i\Rgt_i} = \{t_1,t_2,t_3,t_4\}$, $\S^{O_i} = \{t_5\}$, $\S^{\Lft_i} = \{t_6,t_7,t_8\}$ and 
$\S^{\Rgt_i} = \{t_9,t_{10},t_{11},t_{12}\}$. The three vertices of triangle $t_l$ are annotated with label $l$.}
\label{fig:ex_restruct_Li}
\end{figure}

\end{proof}

\paragraph*{Proof of the L-reduction}
We are now ready to prove the main lemma (recall that $f$ is the
reduction from {\sc Max 2-SAT(3)} to \pb$^{D_M}$ described in
Section~\ref{subsec:reduction2}), and also the main theorem of the section.
\begin{lemma}\label{lem:Lreducv2}
Let $\cal{F}$ be an instance of {\sc Max 2-SAT(3)}. For any $k$, there
exists an assignment $a$ of $\cal{F}$ satisfying at least $k$ clauses
if and only if there exists a solution $\S$ of $f(\cal{F})$ with $|\S|
\geq 11n+m+k$, where $n$ and $m$ are respectively is the number of
variables and clauses in $\cal{F}$.  Moreover, in the $\Leftarrow$
direction, assignment $a$ can be computed from $\S$ in polynomial
time.
\end{lemma}

\begin{proof}
For any $i \in [n]$, let $A_i =P_i$ if $x_i$ is set to true in $a$,
and $A_i=\overline{P_i}$ otherwise. We first add to $\S$ the set
$\cup_{i \in [n]}A_i$.  Then, let $\{C_{j_l}, l \in [k]\}$ be $k$
clauses satisfied by $a$.  For any $l \in [k]$, let $i_l$ be the index
of a literal satisfying $C_{j_l}$, let $x \in [2]$ such that
$c^x_{j_l}$ corresponds to this literal, and let $Z_l =
\{x^2_{i_l},x^{'2}_{i_l}\}$ if literal $i_l$ is positive, and $Z_l =
\{\overline{x^2_{i_l}}\}$ otherwise. For any
$j \in [m]$, if $j=i_l$ for some $l$ (meaning that $j$ corresponds to
a satisfied clause), we add to $\S$ the triangle in $Q^{3-x}_j$, and
otherwise we arbitrarily add the triangle $Q^1_j$.  Finally, for any
$l \in [k]$ we add to $\S$ triangle $t_l =
(y_l,\theta_{j_l},c^x_{j_l})$ where $y_l \in Z_l$ is such that $y_l$
is not already used in another triangle. Notice that such an $y_l$
always exists as triangles of $A_{i}, i \in [n]$ do not intersect
$Z_l$ (by definition of the $A_i$), and as there are at most two
clauses that are true due to positive literal, and one clause that is
true due to a negative literal.
Thus, $\S$ has $11n+m+k$ vertex disjoint triangles.

Conversely, let $\S$ a solution of $f(\cal{F})$ with $|\S| \geq
11n+m+k$. By Lemma~\ref{lem:goodpatternv2} we can construct in
polynomial time a solution $\S'$ from $\S$ such that $|\S'| \ge |\S|$,
$\S'$ only contains outer, variable or clause inner triangles, for
each $j \in [m]$ there exists $x \in [2]$ such that $I^{'K}_j=Q^x_j$,
and for each $i\in[n], I^{'L}_i = P_i$ or $I^{'L}_i =
\overline{P_i}$. This implies that the $k' \ge k$ remaining triangles
must be outer triangles. Let $\{t'_l, l \in [k']\}$ be these $k'$
outer triangles with $t'_l = (y_l,\theta_{j_l},c^{x_l}_{j_l})$
Let us define the following assignation $a$: for each $i\in[n]$, we set $x_i$ to true if $I^{'L}_i = P_i$, and false otherwise. 
This implies that $a$ satisfies at least clauses $\{C_{j_l}, l \in [k']\}$.
\end{proof}
\begin{theorem}\label{thm:apxhv2}
\pb$^{D_M}$ is {\sf APX}-hard, and thus does not admit a {\sf PTAS} unless $P={\sf NP}$.
\end{theorem}
\begin{proof}
Let us check that Lemma~\ref{lem:Lreducv2} implies a $L$-reduction (whose definition is recalled in Definition~\ref{def:L} of appendix).
Let $OPT_1$ (resp. $OPT_2$) be the optimal value of $\cal{F}$ (resp. $f(\cal{F})$).
Notice that Lemma~\ref{lem:Lreducv2} implies that $OPT_2 = OPT_1+11n+m$.
It is known that $OPT_1 \ge \frac{3}{4}m$ (where $m$ is the number of clauses of ${\cal{F}}$). As $n\le m$ (each variable has at least one positive and one negative occurrence), 
we get $OPT_2 = OPT_1+11n+m \le \alpha OPT_1$ for an appropriate constant $\alpha$,  and thus point $(a)$ of the definition is verified.
Then, given a solution $\S'$ of $f(\cal{F})$, according to Lemma~\ref{lem:Lreducv2} we can construct in polynomial time an assignment $a$  satisfying $c(a)$ clauses with $c(a) \ge \S' - 11n-m$.
Thus, the inequality $(b)$ of Definition~\ref{def:L} with $\beta=1$ becomes  $OPT_1-c(a) \le OPT_2 - \S' = OPT_1+11n+m - \S'$, which is true.
\end{proof}

Reduction of Theorem~\ref{thm:apxhv2} does not imply the {\sf
  NP}-hardness of \perfectpb~as there remain some unused vertices.
However, it is straightforward to adapt the reduction by adding
backward arcs whose head (resp. tail) are before (resp. after) $\T$ to
consume the remaining vertices. This leads to the following result.

\begin{theorem}\label{thm:nphperfectv2dm}
 \perfectpb$^{D_M}$ is {\sf NP}-hard. 
\end{theorem}

\begin{proof}
%\fixme{ref de max 2 sat : cycle packing {\sf APX}-hard dans DELTA = 3 et {\sf NP}
%  hard dans planaire + voir ceux qui le cite + ref de max 2 sat (3)}
Let $({\cal F},k)$ be an instance of the decision problem of
$MAX-2-SAT(3)$ and let $\T = f({\cal F})$ be the tournament defined in
Section~\ref{subsec:reduction2}.  Recall that we have $\T = LK$.  Let
$N = |V(T)| = 35n+5m$, $x^* = 33n+3m+3k$ and $n' = N-x^*$.  We now
define $\T'$ by adding $2n'$ new vertices in $\T$ as follows: $V(\T')
= R_1V(\T)R_2$ with $R_i = \{r_i^l, l \in [n']\}$.  We add to
$\bA{A}(\T')$ the set of arcs $R=\{(r_2^lr_1^l), l \in [n']\}$ which
are called the dummy arcs.
We say that a triangle $t=(u,v,w)$ is dummy iff $(wu) in R$ and $v \in V(\T)$.
Let us prove that there are at least $k$ clauses satisfiable in $\cal{F}$ iff there exists a perfect packing in $\T'$.

$\Rightarrow$\\ Given an assignement satisfying $k$ clause we define a
solution $\S$ with $V(\S) \subseteq V(\T)$ as in
Lemma~\ref{lem:Lreducv2} (triangles of $P_i$ or $\overline{P_i}$ for
each $i \in [n]$, a triangle $Q^x_j$ for each $j \in [m]$, and an
outer triangle $t_l$ with $l \in [k]$ for each satisfied clause. We
have $|\S|=11n+m+k$. This implies that $|V(\T) \setminus V(\S)|=n'$,
and thus we use $n'$ remaining vertices of $V(\T)$ by adding to $\S$
$n'$ dummy triangles.

$\Leftarrow$\\ Let $\S'$ be a perfect packing of $\T'$.  Let $\S = \{t
\in \S' : V(t) \subseteq V(\T)\}$.  Let $X = V(\T) \setminus
V(\S)$. As $\S'$ is a perfect packing of $\T'$, vertices of $X$ must
be used by $|X|$ dummy triangles of $\S'$, implying $|X| \le n'$ and
$|\S| \ge 11n+m+k$.  As $\S$ is set of vertex disjoint triangles of
$\T$ of size at least $11n+m+k$, this implies by
Lemma~\ref{lem:Lreducv2} that at least $k$ clauses are satisfiable in
$\cal{F}$.

\end{proof}

To establish the kernel lower bound of Section~4, we also need the {\sf NP}-hardness
of \perfectpb~where instances have a slightly simpler structure (to
the price of losing the property that there exists a FAS which is a
matching).

 \begin{theorem}\label{thm:nphperfectv2}
  \perfectpb~remains {\sf NP}-hard even restricted to tournament $\T$ admitting the following linear ordering.
 \begin{itemize}
 \item $\T = LK$ where $L$ and $K$ are two tournaments
 \item tournaments $L$ and $K$ are "fixed":
 \begin{itemize}
 \item $K = K_1\dots K_m$ for some $m$, where for each $j \in [m]$ we have $V(K_j) = (\theta_j,c_j)$
  \item $L=R_1L_1 \dots L_n R_2$, where each $L_i$ has is a copy of the variable gadget of Section~\ref{subsec:reduction2}, $R_i = \{r_i^l, l \in [n']\}$  where $n'=2n-m$, and in addition $\bA{L}$ also contains $R =\{(r_2^lr_1^l), l \in [n']\}$ which are called the dummy arcs.
 \end{itemize}
 \fixme{je commente sinon ça compile pas mais on a besoin de ctte propriete !for any $a \in \bA{A}(\T)$, $V(a) \cap V(K) \neq \emptyset$ implies $a=vc_j$ for $v \in L$ (there are no backward arc included in $K$, and all $\theta_j$ have degree $(0,0)$)}
 \end{itemize}
 \end{theorem}

\begin{proof}
We adapt the reduction of Section~\ref{subsec:reduction2}, reducing now from 3-SAT(3) instead of MAX 2-SAT(3).
Given $\cal{F}$ be an instance of {\sc 3-SAT(3)} with $n$ variables $\{x_i\}$ nd $m$ clauses $\{C_j\}$.
For each variable $x_i$ with $i \in [n]$, we create a tournament $L_i$ exactly as in Section~\ref{subsec:reduction2} and we define $L=L_1 \dots L_n$.
For each clause $C_j$ with $j \in [m]$, we create a tournament $K_j$ with $V(K_j) = (\theta_j,c_j)$, and we define $K = K_1\dots K_m$. 
Let us now define $\T = LK$. Now, we add to $\bA{A}(\T)$ the following backward arcs from $V(K)$ to $V(L)$ (again, we follow the construction of Section~\ref{subsec:reduction2} except
that now each $c_j$ has degree $(3,0)$). If $C_j = l_{i_1} \vee l_{i_2} \vee l_{i_3}$ is a clause in $\cal{F}$ then we add the arcs 
$c_jv_{i_1}, c_jv_{i_2}, c_jv_{i_3}$ where $v_{i_c}$ is the vertex in $\{x_{i_c}^2,x_{i_c}^{'2},\overline{x_{i_c}^2}\}$ corresponding to $l_{i_c}$: if $l_{i_c}$ is a positive occurrence of variable $i_c$ we chose
$v_{i_c} \in \{x_{i_c}^2,x_{i_c}^{'2}\}$, otherwise we chose $v_{i_c} = \overline{x_{i_c}^2}$. Moreover, we chose vertices $v_{i_c}$ in such a way that for any $i \in [n]$, for each $v \in \{x_i^2,x_i^{'2},\overline{x_i^2}\}$ there exists a unique arc $a \in \bA{A}(\T)$ such that $h(a)=v$. This is always possible as each variable has at most $2$ positive occurrences and $1$ negative one.

Finally, we add $2n'$ new vertices in $\T$ as follows: $V(T) = R_1V(L)R_2V(K)$, $R_i = \{r_i^l, l \in [n']\}$  where $n'=2n-m$.
We add to $\bA{A}(\T)$ the set of arcs $R =\{(r_2^lr_1^l), l \in [n']\}$ which are called the dummy arcs.
Notice that $\T$ satisfies the claimed structure (defining the left part as $R_1LR_2$ and not only $L$).
We define an outer and variable inner triangle as in Section~\ref{sec:approx} (there are no more clause inner triangle), and in addition we say that a triangle $t=(u,v,w)$ is dummy iff $(wu) \in R$ and $v \in V(L)$.
Let us prove that there is an assignment satisfying the $m$ clauses of $\cal{F}$ iff $\T$ has a perfect packing.

$\Rightarrow$\\
Given an assignment satisfying the $m$ clauses we define a solution $\S$ containing only outer, variable inner and dummy triangles. 
The variable inner triangle are defined as in Lemma~\ref{lem:Lreducv2} (triangles of $P_i$ or $\overline{P_i}$ for each $i \in [n]$).
For each clause $j \in [m]$ satisfied by a literal $l_{i_x}$ we create an outer triangle $(v_{i_x},\theta_j,c_j)$.
It remains now $2n-m=n'$ vertices of $L$, that we use by adding $n'$ dummy triangles to $\S$.

$\Leftarrow$\\
Let $\S$ be a perfect packing of $\T'$.
Notice that restructuration lemmas of Section~\ref{sec:approx} do not directly remain true because of the dummy arcs. However, we can adapt in a straightforward manner arguments of these lemmas,
using the fact that $\S$ is even a perfect packing.
Given a solution $\S$, we define as in Section~\ref{sec:approx} set $I^{L}_i =\{t \in \S : V(t) \subseteq V(L_i)$, 
$I^{L} = \cup_{i \in [n]} I^L_i$, $O = \{t \in \S \mbox{ $t$ is an outer triangle }\}$, and $D = \{t \in \S \mbox{ $t$ is a dummy triangle }\}$.
Again, we do not claim (at this point) that $\S$ does not contain other triangles. Given any perfect packing $\S$ of $\T$, we can prove the following properties.
\begin{itemize}
\item $\S$ must contain exactly $m$ outer triangles ($|O| =m$). Indeed, for any $j$ from $m$ to $1$, the only way to use $\theta_j$ is to create 
an outer triangle $(u_j,\theta_j,c_j)$. This implies that triangles of $O$ consume exactly $m$ disjoint vertices in $L$.
\item for any $i \in [n]$, we must have $|I^L_i|=11$.
Indeed, let $x$ be the number of vertices of $L$ used in $\S$ (as $\S$ is a perfect packing we know that $x=|L|=35n$).
The only triangles of $\S$ that can use a vertex of $L$ are the outer, the variable inner and the dummy triangles, implying $x \le (\sum_{i \in [n]}|I^L_i|)+m+n'$
as $|D| \le n'$. As $|V(L_i)| = 35$ we have $|I^L_i| \le 11$ and thus we must have $|I^L_i|=11$ for any $i$.
\end{itemize}

\fixme{c'est ici où c'est un peu crado}
Let us now consider the tournament $\T_0 = \T[V(\T) \setminus V(R)]$ without the dummy arcs, and $\S_0 = \{t \in \S : V(t) \subseteq V(\T_0)\}$.
We adapt in a straightforward way the notion of variable inner and outer triangle in $\T_0$. 
Observe that the variable inner and outer triangles of $\S$ and $\S_0$ are the same, and thus are both denoted respectively $I^L_i$ and $\S^{O_i}$.
In particular, $\S_0$ still contains $m$ outer triangle of $\T_0$.
Now we simply apply proof of Lemma~\ref{lem:goodpatternv2} on $\S_0$. More precisely, Lemma~\ref{lem:goodpatternv2} restructures $\S_0$ into a solution $\S_0'$ with $\S_0' = (\S_0 \setminus (I^L_i \cup \S^{O_i})) \cup (I^{'L}_i \cup \S^{'O_i})$, 
where $I^{'L}_i$ and $\S^{'O_i}$ satisfy properties $(\triangle)$. In particular, as $|I^L_i|=|I^{'L}_i|=11$, $ \triangle_3$ implies that $|\S_0^{'O_i}| \ge |\S_0^{O_i}|$, 
and thus that $|\S_0^{'O}| \ge |\S_0^{O}| = m$. Thus, $\S'_0$ satisfies $I^L_i = P_i$ or $I^L_i = \overline{P_i}$ for any $i$, and has $m$ outer triangles. We can now define 
 as in Lemma~\ref{lem:Lreducv2} from $\S'_0$ an assignment satisfying the $m$ clauses.

\end{proof}

\subsection{$(1+\frac{6}{c-1})$-approximation when backward arcs have large minspan}
Given a set of pairwise distinct pairs $D$ and an integer $c$, we denote by \pb$^D_{\ge c}$ the problem \pb$^D$ restricted to tournaments such that 
there exists a linear representation of minspan at least $c$ and where $d(v) \in D$ for all $v$.
In all this section we consider an instance $\T$ of \pb$^{D_M}_{\ge c}$ with a given linear ordering $(V(\T),\bA{A}(\T))$ of minspan at least $c$ and whose degrees belong to $D_M$.
The motivation for studying the approximability of this special case comes from the situation of MAX-SAT(c) where the approximability becomes easier as $c$ grows, as the derandomized uniform assignment provides a $\frac{2^c}{2^c-1}$ approximation algorithm. Somehow, one could claim that MAX-SAT(c) becomes easy to approximate for large $c$ as there many ways to satisfy a given clause. 
As the same intuition applies for tournament admitting an ordering with large minspan (as there are $c-1$ different ways to use a given backward in a triangle), our objective was to 
find a polynomial approximation algorithm whose ratio tends to $1$ when $c$ increases.

Let us now define algorithm $\ALG$.
We define a bipartite graph $G = (V_1,V_2,E)$ with $V_1 = \{v^1_{a} : a \in \bA{A}(\T)\}$ and
$V_2 =\{v^2_l : v_l \in V_{(0,0)}\}$. 
Thus, to each backward arc we associate a vertex in $V_1$ and to each vertex $v_l$ with $d(v_l) = (0,0)$ we associate a vertex in $V_2$. 
Then, $\{v^1_{a},v^2_l\} \in E$ iff $(h(a),v_l,t(a))$ is a triangle in $\T$.

In phase $1$, $\ALG$ computes a maximum matching $M = \{e_l, l \in [|M|]\}$ in $G$. For every $e_l = \{v^1_{ij},v^2_l\} \in M$ create a triangle $t^1_l = (v_j,v_l,v_i)$.
Let $S^1 = \{t^1_l, l \in [|M|]\}$. Notice that triangles of $S^1$ are vertex disjoint. Let us now turn to phase $2$. 
Let $\T^2$ be the tournament $\T$ where we removed all vertices $V(S^1)$. 
Let $(V(\T^2),\bA{A}(\T^2))$ be the linear ordering of $\T^2$ obtained by
removing $V(S^1)$ in $(V(\T),\bA{A}(\T))$. %\fixme{dans T 2 le span est plus forcement >=c }
We say that three distinct backward edges $\{a_1,a_2,a_3\} \subseteq \bA{A}(\T^2)$ can be packed into triangles $t_1$ and $t_2$ iff $V(\{t_1,t_2\}) = V(\{a_1,a_2,a_3\})$ and the $t_i$ are vertex disjoint.
For example, if $h(a_1) < h(a_2) < t(a_1) < h(a_3) < t(a_2) < t(a_3)$, then $\{a_1,a_2,a_3\}$ can be packed into $(h(a_1),h(a_2),t(a_1))$ and $(h(a_3),t(a_2),t(a_3))$  (recall that 
when $\bA{A}(\T)$ form a matching, $(u,v,w)$ is triangle iff $wu \in \bA{A}(\T)$ and $u<v<w$), and if $h(a_1) < h(a_2) < t(a_2) < h(a_3) < t(a_3) < t(a_1)$, then $\{a_1,a_2,a_3\}$ cannot be packed into two triangles. In phase $2$, while it is possible, $\ALG$ finds a triplet of arcs of $Y \subseteq \bA{A}(\T^2)$ that can be packed into triangles, create the two corresponding triangles, and remove $V(Y)$.
Let $S^2$ be the triangle created in phase $2$ and let $S = S^1 \cup S^2$. 

\begin{observation}\label{obs:arc}
For any $a \in \bA{A}(\T)$, either $V(a) \subseteq V(S)$ or $V(a) \cap V(S) = \emptyset$. Equivalently, no backward arc has one endpoint in $V(S)$ and the other outside $V(S)$.
\end{observation}
According to Observation~\ref{obs:arc}, we can partition $\bA{A}(\T) = \bA{A}_0 \cup \bA{A}_1 \cup \bA{A}_2$, where for $i \in \{1,2\}, $ $\bA{A}^i = \{a \in \bA{A}(\T) : V(a) \subseteq V(S^i)$ is the set of arcs used in phase $i$, 
and $\bA{A}_0 =_{def} \{b_i, i \in [x] \}$ are the remaining unused arcs. Let $\bA{A}_\ALG = \bA{A}_1 \cup \bA{A}_2$, $m_i = |\bA{A}_i|$, $m = m_0+m_1+m_2$ and $m_{\ALG} = m_1+m_2$ the number of arcs (entirely) consumed by $\ALG$.
To prove the $1+f(\frac{6}{c-1})$ desired approximation ratio, we will first prove in Lemma~\ref{lemma:numberarcs} that $\ALG$ uses at most all the arcs ($m_A \ge (1-\epsilon(c))m$), and in Theorem~\ref{thm:approxc} that the number of triangles made with these arcs is "optimal". Notice that the latter condition is mandatory as if $\ALG$ used its $m_\ALG$ arcs to only create $\frac{2}{3}(m_\ALG)$ triangles in phase 2 
instead of creating $m' \approx m_\ALG$ triangle with $m'$ backward arcs and $m'$ vertices of degree $(0,0)$, we would have a $\frac{3}{2}$ approximation ratio.

\begin{lemma}\label{lemma:numberarcs}
For any $c\ge 2$, $m_\ALG \ge (1-\frac{6}{c+5})m$
\end{lemma}

\begin{proof}
In all this proof, the span $s(a)$ is always considered in the initial input $\T$, and not in $\T^2$.
For any $i \in [x]$, let us associate to each $b_i \in \bA{A}_0$ a set $B_i \subseteq \bA{A}_\ALG$ defined as follows (see Figure~\ref{fig:B_i} for an example).
Let $b_j \in \bA{A}_0$ such that $s(b_j) \subseteq s(b_i)$ and there does not exist a $b_k \in \bA{A}_0$ such that
$s(b_k)$ included in $s(b_j)$ (we may have $b_j = b_i$).
Let $Z = V(\bA{A}_0) \cap s(b_j)$. Notice that $|Z| \le 1$, meaning that there is at most one endpoint of a $b_l, l\neq j$ in $s(b_j)$, as otherwise we would be three arcs in $\bA{A}_0$ that could be packed in two triangles. If there exists $a \in \bA{A}_{\ALG}$ with $s(a) \subseteq s(b_j)$ we define $a_0 = a$, and otherwise we define $a_0 = b_j$.
Now, let $v \in s(a_0) \setminus Z$.
 Observe that $V(\T)$ is partitioned into $V(\bA{A}_0) \cup V(\bA{A}_{\ALG}) \cup V_{(0,0)}$. If $v \in V_{(0,0)}$, then there exists $t^1_l = (u,v,w)$ with $wu \in \bA{A}_1$ (as otherwise the matching in phase 1 would not be maximal and we could add $b_j$ and $v$), 
and we add $wu$ to $B_i$. Otherwise, $v \in V(a)$ with $a \in \bA{A}_{\ALG}$ (this arcs could have been used in phase $1$ or phase $2$), and we add $a$ to $B_i$. Notice that as $a_0$ does not properly contains another arc of $\bA{A}_{\ALG}$, all the added arcs are pairwise distinct, and thus $|B_i| = |s(a_0) \setminus Z| \ge c-1$. 

\begin{figure}[!htbp]
\begin{center}
\includegraphics[width=0.75\textwidth]{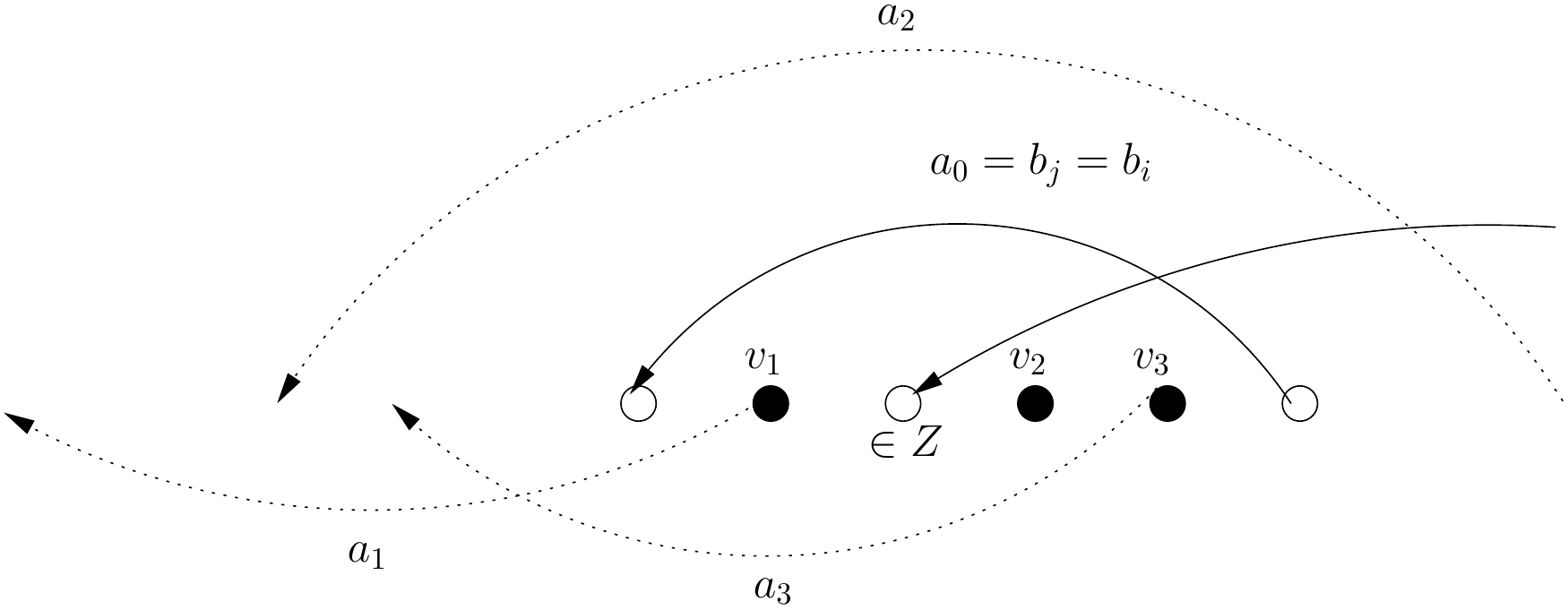}
\end{center}
\caption{On this example white vertices represent $V(\T) \setminus V(S)$ (vertices not used by $\ALG$), and black ones represent $V(S)$. In this case 
we have $B_i = \{a_l, l \in [3]\}$. Indeed, each $v_l \in s(a_0) \setminus Z$, for $l \in [3]$, brings $a_l$ in $B_i$. In particular $v_2 \in V_{(0,0)}$ and was used with $a_2$ to create
a triangle in phase 1.}
\label{fig:B_i}
\end{figure}

\begin{figure}[!htbp]
\begin{center}
\includegraphics[width=0.75\textwidth]{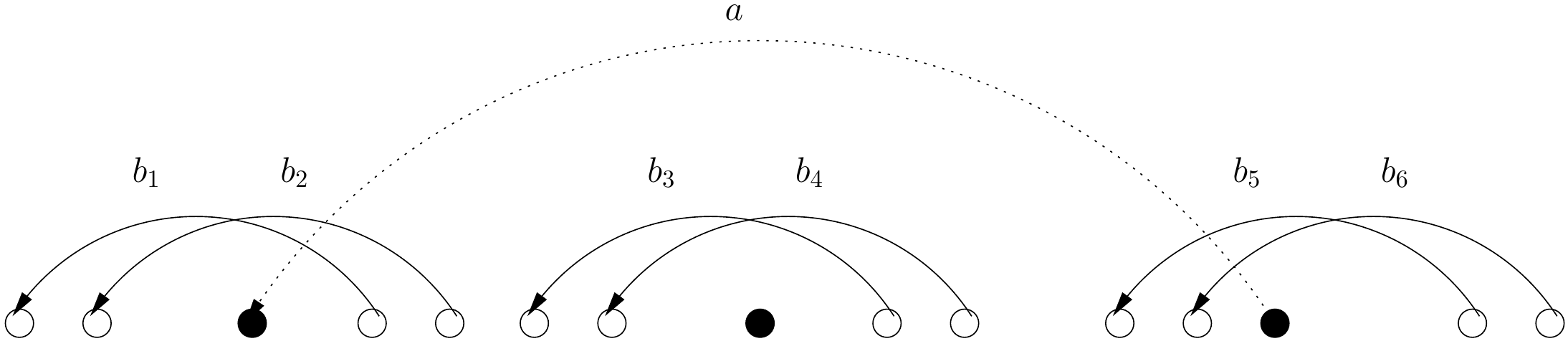}
\end{center}
\caption{Example where $|B(a)|=6$ for $a \in \bA{A}_\ALG$, where $B(a)=\{b_l, l \in [6]\}$.}
\label{fig:B_i_tight}
\end{figure}

Given $a \in \bA{A}_{\ALG}$, let $B(a) = \{B_i, a \in B_i\}$. Let us
prove that $|B(a)| \le 6$ for any $a \in \bA{A}_{\ALG}$.  For any $v
\in V(S)$, let $d_B(v) = |\{b_i : v \in s(b_i)\}|$. Observe that
$d_B(v) \le 2$, as otherwise any triplet of arcs containing $v$ in
their span could be packed into two triangles (there are only $6$
cases to check according to the $3!$ possible ordering of the tail of
these $3$ arcs).
For any $a \in \bA{A}_1$, let $V'(a) = V(t^a)$ where $t^a \in S$ is the triangle containing $a$, 
and for any $a \in A_2$, let $V'(a) = V(a)$.
Observe that by definition of the $B_i$, $a \in B_i$ implies that $b_i$ contributes to the degree $d_B(v)$ for a $v \in V'(a)$. As in particular $d_B(v)$ for any $v \in V'(a)$, 
this implies by pigeonhole principle that $|B(a)| \le 6$ (notice that this bound is tight as depicted Figure~\ref{fig:B_i_tight}).
Thus, if we consider the bipartite graph with vertex set $(\bA{A}_0,\bA{A}_{\ALG})$ and an edge between $b_i \in \bA{A}_0$ and $a \in \bA{A}_{\ALG}$ iff $a \in B_i$, the number of edges $x$ of this graph satisfies
$|\bA{A}_0|(c-1) \le x \le 6|\bA{A}_{\ALG}|$, implying the desired inequality as $m_\ALG = m - m_0$.
\end{proof}

\begin{theorem}\label{thm:approxc}
For any $c \ge 2$, $\ALG$ is a polynomial $(1+\frac{6}{c-1})$ approximation algorithm for \pb$^{D_M}_{\ge c}$.
\end{theorem}
\begin{proof}
Let $OPT$ be an optimal solution. Let us define set $OPT_i \subseteq OPT$ and $\bA{A}^*_i \subseteq \bA{A}(\T)$ as follows.
Let $t=(u,v,w) \in OPT$. As the FAS of the instance is a matching, we know that $wu \in \bA{A}(\T)$ as we cannot have a triangle with two backward arcs.
If $d(v)=(0,0)$ then we add $t$ to $OPT_1$ and $wu$ to $\bA{A}^*_1$.
Otherwise, let $v'$ be the other endpoint of the unique arc $a$ containing $v$. If $v' \notin V(OPT)$, then we add $t$ to $OPT_3$ and $\{wu,a\}$ to $\bA{A}^*_3$.
Otherwise, let $t' \in OPT$ such that $v' \in V(t')$. As the FAS of the instance is a matching we know that $v'$ is the middle point of $t'$, or more formally that 
$t' = (u',v',w')$ with $u'w' \in \bA{A}(\T)$. We add $\{t,t'\}$ to $OPT_2$ and $\{wu,a,w'u'\}$ to $\bA{A}^*_2$. Notice that the $OPT_i$ form a partition of $OPT$, and that the $\bA{A}^*_i$ have pairwise empty intersection, implying $|\bA{A}^*_1|+|\bA{A}^*_2|+|\bA{A}^*_3| \le m$. Notice also that as triangles of $OPT_1$ correspond to a matching of size $|OPT_1|$ in the bipartite graph defined in phase $1$ of algorithm $\ALG$, we have $|OPT_1|=|\bA{A}^*_1| \le |\bA{A}_1|$.

Putting pieces together we get (recall that $S$ is the solution computed by $\ALG$): $|OPT| = |OPT_1|+|OPT_2|+|OPT_3| = |\bA{A}^*_1|+\frac{2}{3}|\bA{A}^*_2|+\frac{1}{2}|\bA{A}^*_3| 
      \le |\bA{A}^*_1|+\frac{2}{3}(|\bA{A}^*_2|+|\bA{A}^*_3|) \le |\bA{A}^*_1|+\frac{2}{3}(m-|\bA{A}^*_1|) \le \frac{1}{3}|\bA{A}_1|+\frac{2}{3}m $ and 
$|S| = |S^1|+|S^2| 
    =   |\bA{A}_1|+\frac{2}{3}|\bA{A}_2|
    \ge |\bA{A}_1|+\frac{2}{3}((1-\frac{6}{c+5})m - |\bA{A}_1|) 
   = \frac{1}{3}|\bA{A}_1|+\frac{2}{3}(1-\frac{6}{c+5})m $
which implies the desired ratio.
\end{proof}

\section{Kernelization} 
\label{sec:kernel}
In all this section we consider the decision problem \pb~parameterized by the size of the solution. Thus, an input is a pair $I=(\T,k)$ and we say that $I$ is
positive iff there exists a set of $k$ vertex disjoint triangles in $\T$.

\subsection{Positive results for sparse instances}

Observe first that the kernel in $\O(k^2)$ vertices for $3$-{\sc Set Packing} of~\cite{abu2009quadratic} directly implies a kernel in $\O(k^2)$ vertices for \pb.
Indeed, given an instance $(\T=(V,A),k)$ of \pb, we create an instance $(I'=(V,C),k)$ of $3$-{\sc Set Packing} by creating an hyperedge $c \in C$ for each triangle of $\T$. 
Then, as the kernel of~\cite{abu2009quadratic} only removes vertices, it outputs an induced instance $(\overline{I'}=I'[V'],k')$ of $I$ with $V' \subseteq V$, and thus this induced instance can be interpreted 
as a subtournament, and the corresponding instance $(\T[V'],k')$ is an equivalent tournament with $\O(k^2)$ vertices.

As shown in the next theorem, as we could expect it is also possible to have kernel bounded by the number of backward arcs.

\begin{theorem}\label{thm:kernelm}
\pb~admits a polynomial kernel with $\O(m)$ vertices, where $m$ is the number of arcs in a given FAS of the input.
\end{theorem}
\begin{proof}
Let $I=(\T,k)$ be an input of the decision problem associated to \pb. Observe first that we can build in polynomial time a linear ordering $\V(\T)$ whose backward arcs $\bA{A}(\T)$ correspond to the given FAS. We will obtain the kernel by removing useless vertices of degree $(0,0)$. 
Let us define a bipartite graph $G = (V_1,V_2,E)$ with $V_1 = \{v^1_{a} : a \in \bA{A}(\T)\}$ and $V_2 =\{v^2_l : v_l \in V_{(0,0)}\}$.
Thus, to each backward arc we associate a vertex in $V_1$ and to each vertex $v_l$ with $d(v_l) = (0,0)$ we associate a vertex in $V_2$. 
Then, $\{v^1_{a},v^2_l\} \in E$ iff $(h(a),v_l,t(a))$ is a triangle in $\T$.
By Hall's theorem, we can in polynomial time partition $V_1$ and $V_2$ into $V_1=A_1 \cup A_2$, $V_2=B_0 \cup B_1 \cup B_2$ such that $N(A_2) \subseteq B_2$, $|B_2| \le |A_2|$, and 
there is a perfect matching between vertices of $A_1$ and $B_1$ ($B_0$ is simply defined by $B_0 = V_2 \setminus (B_1 \cup B_2)$). 

For any $i, 0 \le  i \le 2$, let $X_i = \{v_l \in V_{(0,0)} : v^2_l \in B_i\}$ be the vertices of $\T$ corresponding to $B_i$. 
Let $V_{\neq(0,0)} = V(\T) \setminus V_{(0,0)}$. Notice that $|V_{\neq(0,0)}| \le 2m$.
We define $\T' = \T[V_{\neq(0,0)} \cup X_1 \cup X_2]$ the sub-tournament obtained from $\T$ by removing vertices of $X_0$, 
and $I' = (\T',k)$. We point out that this definition of $\T'$ is similar to the final step of the kernel of~\cite{abu2009quadratic} as our partition of $V_1$ and $V_2$ (more precisely
$(A_1,B_0 \cup B_1)$) corresponds in fact to the crown decomposition of~\cite{abu2009quadratic}. Observe that $|V(\T')| \le 2m+|A_1|+|A_2| \le 3m$, implying the desired bound of the number of vertices of the kernel.

It remains to prove that $I$ and $I'$ are equivalent. Let $k \in \mathbb{N}$, and let us prove that there exists a solution $\S$ of $\T$ with $|\S| \ge k$ iff 
there exists a solution $\S'$ of $\T'$ with $|\S'| \ge k$. Observe that the $\Leftarrow$ direction is obvious as $\T'$ is a subtournament of $\T$. Let us now prove the $\Rightarrow$ direction.
Let $\S$ be a solution of $\T$ with $|\S| \ge k$. Let $\S = \S_{(0,0)} \cup \S_1$ with $\S_{(0,0)} = \{t \in \S : t=(h(a),v,t(a))\mbox{ with } v \in V_{(0,0)}, a \in \bA{A}(\T)\}$ 
and $\S_1 = \S \setminus \S_{(0,0)}$. Observe that $V(\S_1) \cap V_{(0,0)} = \emptyset$, implying $V(\S_1) \subseteq V_{\neq(0,0)}$.
For any $i \in [2]$, let $\S^i_{(0,0)} = \{t \in \S_{(0,0)} : t=(h(a),v,t(a))\mbox{ with } v \in V_{(0,0)}, v^1_a \in A_i\}$ be a partition of $\S_{(0,0)}$.
We define $\S' = \S_1 \cup \S^2_{(0,0)} \cup \S^{'1}_{(0,0)}$, where $\S^{'1}_{(0,0)}$ is defined as follows. For any $v^1_a \in A_1$, let $v^2_{\mu(a)} \in B_1$ be the vertex associated to $v^1_a$ 
in the $(A_1,B_1)$ matching. To any triangle $t=(h(a),v,t(a)) \in \S^1_{(0,0)}$ we associate a triangle $f(t)=(h(a),v_{\mu(a)},t(a)) \in \S^{'1}_{(0,0)}$, where by definition $v_{\mu(a)} \in X_1$.
For the sake of uniformity we also say that any $t \in \S_1 \cup \S^2_{(0,0)}$ is associated to $f(t)=t$.

Let us now verify that triangles of $\S'$ are vertex disjoint by verifying that triangles of $\S^{'1}_{(0,0)}$ do not intersect another triangle of $\S'$.
Let $f(t)=(h(a),v_{\mu(a)},t(a)) \in \S^{'1}_{(0,0)}$. Observe that $h(a)$ and $t(a)$ cannot belong to any other triangle $f(t')$ of $\S'$ as for any $f(t'') \in \S'$, $V(f(t'')) \cap V_{\neq(0,0)} = V(t'') \cap V_{\neq(0,0)}$ 
(remember that we use the same notation $V_{\neq(0,0)}$ to denote vertices of degree $(0,0)$ in $\T$ and $\T'$).
Let us now consider $v_{\mu(a)}$. For any $f(t') \in \S_1$, as $V(f(t')) \cap V_{(0,0)} = \emptyset$ we have $v_{\mu(a)} \notin V(f(t'))$. 
For any $f(t')=(h(a'),v_l,t(a')) \in \S^{2}_{(0,0)}$, we know by definition that $v^1_{a'} \in A_2$, implying that $v^2_l \in B_2$ (and $v_l \in X_2$) as $N(A_2) \subseteq B_2$ and thus that $v_l \neq v_{\mu(a)}$.
Finally, for any $f(t')=(h(a'),v_l,t(a')) \in \S^{'1}_{(0,0)}$, we know that $v_l = v_{\mu(a')}$, where $a \neq a'$, leading to $v_l \neq v_{\mu(a)}$ as $\mu$ is a matching.
\end{proof}

Using the previous result we can provide a $\O(k)$ vertices kernel for
\pb~restricted to sparse tournaments. 

\begin{theorem}\label{thm:kernel-for-sparse}
\pb~restricted to sparse tournaments admits a polynomial kernel with
$\O(k)$ vertices, where $k$ is the size of the solution.
\end{theorem}

\begin{proof}
Let $I=(\T,k)$ be an input of the decision problem associated to
\pb~such that $\T$ is a sparse tournament.  We say that an arc $a$ is
a \emph{consecutive backward arc} of $\V(\T)$ if it is a backward arc
of $\T$ and $a=v_{i+1}v_i$ with $v_i$ and $v_{i+1}$ being consecutive
in $\V(\T)$.  If $\T$ admits a consecutive backward arc $v_iv_{i+1}$
then we can exchange $v_i$ and $v_{i+1}$ in $\T$. The backward arcs of
the new linear ordering is exactly $\bA{A}(\T)\setminus v_{i+1}v_i$
and so is still a matching. Hence we can assume that $\T$ does not
contain any consecutive backward arc.  Now if $|\bA{A}(\T)|< 5k$ then
by Theorem~\ref{thm:kernelm} we have a kernel with $\O(k)$
vertices. Otherwise, if $|\bA{A}(\T)|\ge 5k$ we will prove that $T$ is
a {\sc yes} instance of \pb . Indeed we can greedily produce a family
of $k$ vertex disjoint triangles in $T$.  For that consider a backward
arc $v_jv_i$ of $\T$, with $i<j$. As $v_jv_i$ is not consecutive there
exists $l$ with $i<l<j$ and we select the triangle $v_iv_jv_l$ and
remove the vertices $v_i$, $v_l$ and $v_j$ from $\T$. Denote by $\T'$
the resulting tournament and let $\V(\T')$ be the order induced by
$\V(\T)$ on $\T'$. So we loose at most 2 backward arcs in $\V(\T')$
($v_jv_i$ and a possible backward arc containing $v_l$) and create at
most 3 consecutive backward arcs by the removing of $v_i$, $v_l$ and
$v_j$. Reducing these consecutive backward arcs as previously, we can
assume that $\V(\T')$ does not contain any consecutive backward arc
and satisfies $|\bA{A}(\T')|\ge |\bA{A}(\T)|-5 \ge 5(k-1)$. Finally
repeating inductively this arguments, we obtain the desired family of
$k$ vertex-disjoint triangles in $\T$, and $\T$ is a {\sc yes}
instance of \pb.
\end{proof}

\subsection{No (generalised) kernel in ${\cal O}(k^{2-\epsilon})$}

In this section we provide an OR-cross composition (see
Definition~\ref{def:orcompo} in Appendix) from \perfectpb~restricted
to instances of Theorem~\ref{thm:nphperfectv2} to \perfectpb~
parameterized by the total number of vertices.

\paragraph*{Definition of the instance selector}
The next lemma build a special tournament, called an \emph{instance
  selector} that will be useful to design the cross composition. 

\begin{lemma}\label{lem:path}
For any $\g =2^{\gp}$ and $\m$ we can construct in polynomial time (in
$\g$ and $\m$) a tournament $\Pg$ such that
\begin{itemize}
\item there exists $\g$ subsets of $\m$ vertices $\Xg{i}=\{x^i_j : j
  \in [\m] \}$, that we call the distinguished set of vertices, such
  that
\begin{itemize}
\item the $\Xg{i}$ have pairwise empty intersection
\item for any $i \in [\g]$, there exists a packing $\S$ of triangles
  of $\Pg$ such that $V(\Pg) \setminus V(\S) = \Xg{i}$ (using this
  packing of $\Pg$ corresponds to select instance $i$)
\item for any packing $\S$ of triangles of $\Pg$ with
  $|V(\S)|=|V(\Pg)|-\m$ there exists $i \in [\g]$ such that $V(\Pg)
  \setminus V(\S) \subseteq \Xg{i}$
\end{itemize}
\item $|V(\Pg)|=\O(\m\g)$.
\end{itemize}
\end{lemma}

\begin{proof}
Let us first describe vertices of $\Pg$.  For any $i \in [\g-1]_0$
(where $[x]_0$ denotes $\{0,\dots,x\}$) let $\Xg{i}=\{x^i_j : j \in
[\m] \}$, and let $X = \cup_{i\in [\g-1]_0}\Xg{i}$. For any $l \in
[\gp-1]_0$, let $\Vg{l}=\{v^l_k,k \in [|\Vg{l}|]\}$ be the vertices of
level $l$ with $|\Vg{l}|= \m \g/2^{l} +2$, and $V=\cup_{l \in
  [\gp-1]_0}\Vg{l}$. Finally, we add a set $\Alg{}=\{\Alg{l} :
l\in[\gp-1]_0\}$ containing one dummy vertex for each level and
finally set $V(\Pg)= X\cup V\cup \Alg{}$.  Observe that
$|V(\Pg)|=\m\g+\sum_{l \in [\gp-1]_0}(|\Vg{l}|+1)=\O(\m\g)$.  Let us
now describe $\V(\Pg)$ and $\bA{A}(\Pg)$ recursively. Let $\Pgc{0}$ be
the tournament such that $\V(\Pgc{0})=(v^0_1, x^0_1, v^0_2, x^1_1,
\dots , v^0_\g, x^{\g-1}_1)$ $(v^0_{\g+1}, x^0_2, \dots, v^0_{2\g},
x^{\g-1}_2)$ $\dots$ $(v^0_{(\m-1)\g+1}$ $,x^0_\m,\dots,
v^0_{\m\g},x^{\g-1}_\m)$ $(v^0_{\m\g+1}, \Alg{1}, v^0_{\m\g+2})$ and
$\bA{A}(\Pgc{0})=Z^0_P$ where $Z^0_P=A^0_P \cup A^{'0}_P$ with
$A^0_P=\{v^0_{k+1}v^0_{k} : k \in [|\Vg{0}|-2]\}$ and $A^{'0}_P =
\{v^0_{|\Vg{0}|}v^0_{|\Vg{0}|-1},v^0_{|\Vg{0}|}v^0_1\}$.

Then, given a tournament $\Pgc{l}$ with $0 \leq l <\gp -1$, we
construct the tournament $\Pgc{l+1}$ such that the vertices of
$\Pgc{l+1}$ are those of $\Pgc{l}$ to which are added the set
$\Vg{l+1}$.  For $j \in [|\Vg{l+1}|-2]$, we add the vertex $v^{l+1}_j$
of $\Vg{l+1}$ just after the vertex $v^l_{2j-1}$ in the order of
$\Pgc{l+1}$, and we for $i \in \{0,1\}$ we add vertex
$v^{l+1}_{|\Vg{l+1}|-i}$ just after $v^{l}_{|\Vg{l}|-i}$.  Similarly,
we add the vertex $\Alg{l+1}$ just after the vertex $\Alg{l}$. The
backward arcs of $\Pgc{l+1}$ are defined by: $\bA{A}(\Pgc{l+1}) =
\bA{A}(\Pgc{l}) \cup Z^{l+1}_P$ where $Z^{l+1}_P=A^{l+1}_P \cup
A^{'l+1}_P$ are called the \emph{arcs of level $l$}, with
$A^{l+1}_P=\{v^{l+1}_{k+1}v^{l+1}_{k} : k \in [|\Vg{l+1}|-2]]\}$ and
  $A^{'l+1}_P=\{v^{l+1}_{|\Vg{l+1}|}v^{l+1}_{|\Vg{l+1}|-1},v^{l+1}_{|\Vg{l+1}|}v^{l+1}_1\}$. We
  can now define our gadget tournament $\Pg$ as the tournament
  corresponding to $\Pgc{\gp-1}$. We refer the reader to Figure~\ref{fig:Pathgadget} where an example of the gadget is depicted, where $\m = 3$ and $\g=4$.
	
	\begin{figure}%
\centering
\includegraphics[width=\textwidth]{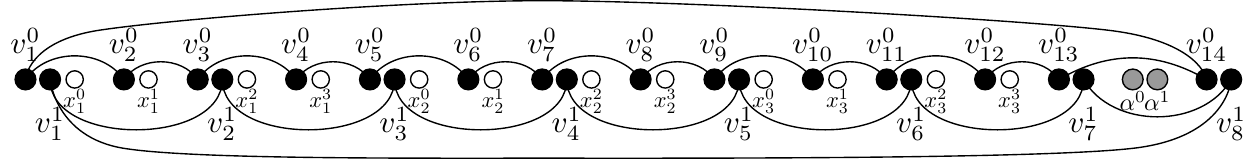}
\caption{An example of the instance selector, where $\m = 3$ and $\g=4$. All depicted arcs are backward arcs.}%
\label{fig:Pathgadget}%
\end{figure}

In all the following given $i \in [\g-1]_0$ we call the last $x$ bits
(resp. the $x^{th}$ bit) $i$ its $x$ right most (resp. the $x^{th}$,
starting from the right) bits in the binary representation of $i$.
Let us now state the following observations. 
\begin{itemize}
\item[$\triangle_1$] The vertices of $X$ have degree $(0,0)$ in $\Pg$.
\item[$\triangle_2$] For any $l \in [\gp-1]_0$, the extremities of the
  arcs of level $l$ are exactly $V^l$ ($V(Z^l_P) = V^l$) and the arcs
  of $Z^l_P$ induce an even circuit on $V^l$.
\item[$\triangle_3$] For any $a \in A^l_P$, the span of $a$ contains
  $2^l$ consecutive vertices of $X$, more precisely $s(a) \cap X =
  \{x^i_j,\dots,x^{i+2^l-1}_j\}$ for $j \in [m]$ and $i$ such that the
  $l-1$ last bits of $i$ are equal to $0$.
\item[$\triangle_4$] There is a unique partition $Z^l_P = Z^{l,0}_P
  \cup Z^{l,1}_P$ such that $|Z^{l,0}_P|=|Z^{l,1}_P|=\mm{l}$, the size
  of a maximum matching of backward arcs in $\Pg[\Vg{l}]$, such that
  each $Z^{l,x}_P$ is a matching (for any $a,a' \in Z^{l,x}_P, V(a)
  \cap V(a') = \emptyset$), and such that $\cup_{a \in Z^{l,x}_P
    \setminus A^{'l}_P} s(a) \cap X$ is the set of all vertices
  $x^i_j$ of $X$ whose $l^{th}$ bit of $i$ is $x$.
\end{itemize}
Now let us first prove that for any $i \in [\g-1]_0$, there exists an
packing $\S$ of $\Pg$ such that $V(\Pg) \setminus V(\S) = \Xg{i}$.
Let $(x_{\gp-1} \dots x_0)$ be the binary representation of $i$.  Let
us define recursively $\S=\cup_{l \in [\gp-1]_0}\S_l$ in the following
way.  We maintain the invariant that for any $l$, the remaining
vertices of $X$ after defining $\cup_{z \in [l]_0}\S_z$ are all the
vertices of $X$ having their $l$ last bits equal to
$(x_{l-1},\dots,x_0)$.  We define $\S_l$ as the $\mm{l}-1$ triangles
$\{(h(a),x_a,t(a), a \in Z^{l,1-x_l}_P) \setminus A^{'l}_P \}$ such
that $x_a$ is the unique remaining vertex of $X$ in $s(a)$ (by
$\triangle_3$ and our invariant of the $\S_{\le l}$, there remains
exactly one vertex in $s(a)$, and by $\triangle_4$ these $\mm{l}-1$
triangles consume all remaining vertices of $X$ whose $l^{th}$ bit is
$1-x_l$), and a last triangle using an arc in $A^{'l}_P$ with
$t=(v^l_{|\Vg{0}|},\Alg{l},v^l_{|\Vg{0}|-1})$ if $x_l = 1$ and
$t=(v^l_{0},\Alg{l},v^0_{|\Vg{0}|})$ otherwise.  Thus, by our
invariant, the remaining vertices of $X$ after defining $\S$ are
exactly $\Xg{i}$.  As $\S$ also consumes $\alpha$ and $V$ we have
$V(\Pg) \setminus V(\S) = \Xg{i}$. Notice that this definition of $\S$
shows that $|V(\Pg)|-m = |V(\S)| = 3\sum_{l \in [\gp-1]_0}\mm{l}$

Let us now prove that for any packing $\S$ of $\Pg$ with
$|V(S)|=|V(\Pg)|-m=3\sum_{l \in [\gp-1]_0}\mm{l}$, there exists $i \in
[\g]$ such that $V(\Pg) \setminus V(\S) \subseteq \Xg{i}$.  Let $t_1,
\dots, t_\mm{}$ be the triangles of $\S$. For any $t_k$ of $\S$, we
associate one backward arc $a_k$ of $t_k$ (if there are two backward
arcs, we pick one arbitrarily).  Let $Z=\{a_k : k \in [|\S|]\}$ and
for every $l \in |\gp-1]_0$ let $Z^l=\{a_k\in A : V(a_k) \subset
  \Vg{l}\}$ the set of the backward arcs which are between two
  vertices of level $l$.  Notice that the $Z^l$ 's form a partition of
  $Z$.  For any $l \in [\gp-1]_0$, we have $|Z^l| \leq \mm{l}$ as two
  arcs of $Z^l$ correspond to two different triangles of $\S$,
  implying that $Z^l$ is a matching.  Furthermore, as
  $|\S|=|Z|=\sum_{l \in [\gp-1]_0}|Z^l|=\mm{} =
  \sum_{l\in[\gp]}{\mm{l}}$, we get the equality $|Z^l| = \mm{l}$ for
  any $l \in [\gp-1]_0$.  This implies that for each $Z^l$ there
  exists $x$ such that $Z^l=Z^{l,x}_P$, implying also that
  $V(Z^l)=\Vg{l}$, and $V(Z)=\Vg{}$.

Let $A^l = Z^l \setminus A^{'l}_P$, $\S^{l} = \{t_k \in \S : a_k \in
A^l\}$.
We can now prove by induction that all the remaining vertices $R_l=X
\setminus V(\cup_{x \in [l]_0} \S^{l})$ have the same $l$ last bits.
Notice that since all vertices of $\Vg{}$ are already used, and as
triangles of $\S^l$ cannot use a dummy vertex in $\alpha$, all
triangles of $\S^l$ must be of the from $(h(a_k),x,t(a_k))$ with $x
\in X$. As $A^l= Z^{l,x}_P \setminus A^{'l}_P$, by $\triangle_4$ we
know that $\cup_{a \in A^l} s(a) \cap X$ contains all the remaining
vertices of $X$, and thus of $R_{l-1}$, whose $l^{th}$ bit is
$x$. Moreover, by $\triangle_3$ we know that for any $a \in A^l$ we
have $|R_{l-1} \cap s(a)| \le 1$ because as $a \in A^l_P$ we know
exactly the structure of $s(a) \cap X$, and if there remain two
vertices in $s(a) \cap X$ then their last $l-1$ last bits would be
different. Thus, as triangles of $\S^l$ remove on vertex in the span
of each $a \in A^l$, they remove all vertices of $R_{l-1}$ whose
$l^{th}$ bit is $x$, implying the desired result.

\end{proof}

\paragraph*{Definition of the reduction}
We suppose given a family of $t$ instances $F=\{\Fb_l, l \in [t]\}$ of
\perfectpb~restricted to instances of Theorem~\ref{thm:nphperfectv2}
where $\Fb_l$ asks if there is a perfect packing in $\T_l=L_lK_l$.  We
chose our equivalence relation in Definition~\ref{def:orcompo} such
that there exist $n$ and $m$ such that for any $l \in [t]$ we have
$|V(L_l)|=n$ and $|V(K_l)|=m$. We can also copy some of the $t$
instances such that $t$ is a square number and $g = \sqrt{t}$ is a
power of two. We reorganize our instances into $F=\{ \Fb_{(p,q)} : 1
\leq p,q \leq g \}$ where $\Fb_{(p,q)}$ asks if there is a perfect
packing in $\T_{(p,q)}=L_pK_q$.  Remember that according to
Theorem~\ref{thm:nphperfectv2}, all the $L_p$ are equals, and all the
$K_q$ are equals.  We point out that the idea of using a problem on
"bipartite" instances to allow encoding $t$ instances on a "meta"
bipartite graph $G=(A,B)$ (with $A=\{A_i, i \in \sqrt{t}\}$, $B=\{B_i,
i \in \sqrt{t}\}$) such that each instance $p,q$ is encoded in the
graph induced by $G[A_i \cup B_i]$ comes from~\cite{dell2014kernelization}. 
We refer the reader to Figure~\ref{fig:Reduc2} which represents the
different parts of the tournament. We define a tournament
$G=\Lb \Mb_G \Lt \Mt_G P_{(n,g)}$, where $\Lb = \Lb_1 \dots \Lb_g$,
$\Mt_G$ is a set of $n$ vertices of degree $(0,0)$, $\Mb_G$ is a set
of $(g-1)n$ vertices of degree $(0,0)$, $\Lt = \Lt_1 \dots \Lt_g$
where each $\Lt_p$ is a set of $n$ vertices, and $P_{(n,g)}$ is a copy
of the instance selector of Lemma~\ref{lem:path}.  Then, for every $p
\in [g]$ we add to $G$ all the possible $n^2$ backward arcs going from
$\Lt_p$ to $\Lb_p$. Finally, for every distinguished set $X^p$ of
$P_{(n,g)}$ (see in Lemma~\ref{lem:path}), we add all the possible
$n^2$ backward arcs from $X^p$ to $\Lt_p$.

Now, in a symmetric way we define a tournament $D=\Kb \Mb_D \Kt \Mt_D
P'_{(m,g)}$, where $\Kb = \Kb_1 \dots \Kb_g$, $\Mt_D$ is a set of $m$
vertices of degree $(0,0)$, $\Mb_D$ is a set of $(g-1)m$ vertices of
degree $(0,0)$, $\Kt = \Kt_1 \dots \Kt_g$ where each $\Kt_q$ is a set
of $m$ vertices, and $P'_{(m,g)}$ is a copy of the instance selector
of Lemma~\ref{lem:path}.  Then, for every $q \in [g]$ we add to $G$
all the $m^2$ possible backward arcs going from $\Kt_p$ to
$\Kb_p$. Finally, for every distinguished set $X^{'q}$ of $P'_{(m,g)}$
we add all the possible $m^2$ backward arcs from $X^{'q}$ to $\Kt_q$.
Finally, we define $\T = GD$. Let us add some backward arcs from $D$
to $G$. For any $p$ and $q$ with $1\leq p, q \leq g$, we add backward
arcs from $\Kb_q$ to $\Lb_p$ such that $\T[\Kb_q\Lb_p]$ corresponds to
$\T_{(p,q)}$. Notice that this is possible as for any fixed $p$, all
the $\T_{(p,q)}, q \in [g]$ have the same left part $\Lb_p$, and the
same goes for any fixed right part.

\begin{figure}%
\centering
\includegraphics[width=\textwidth]{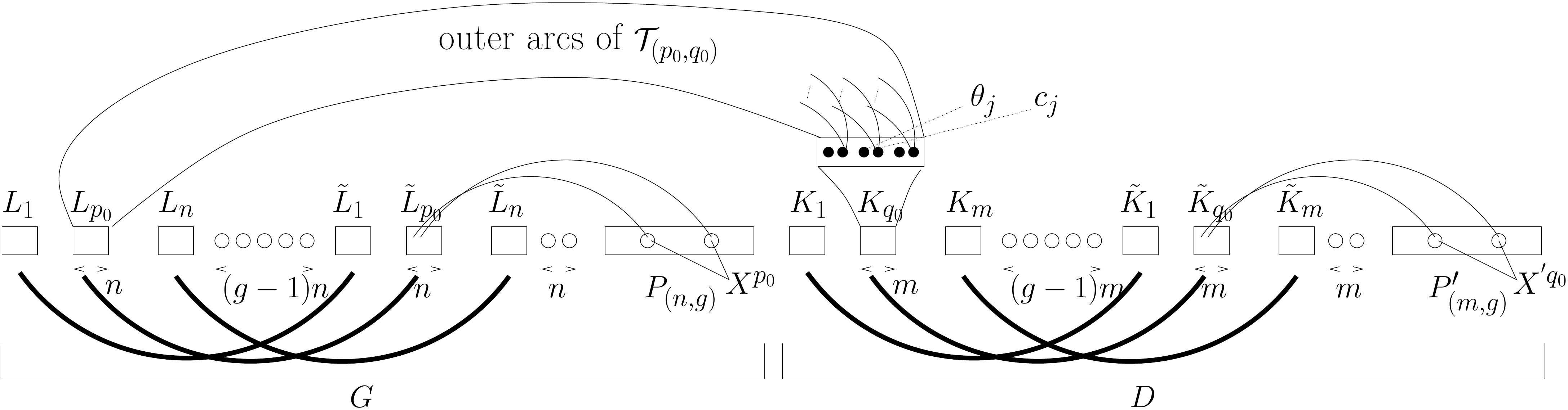}
\caption{A example of the weak composition. All depicted arcs are backward arcs. Bold arcs represents the $n^2$ (or $m^2$) possible arcs between the two groups.}%
\label{fig:Reduc2}%
\end{figure}

\paragraph*{Restructuration lemmas}

Given a set of triangles $\S$ we define $\S_{\subseteq P'}=\{t \in \S
| V(t) \subseteq P'_{(m,g)}\}$, $\S_{\subseteq P}=\{t \in \S : V(t)
\subseteq P_{(n,g)}\}$, $\S_{\Mt_D}= \{t \in \S : V(t) \mbox{
  intersects $\Kt$, $\Mt_D$ and $P'_{m,g}$}\}$, $\S_{\Mb_D}= \{t \in
\S : V(t) \mbox{ intersects $\Kb$, $\Mb_D$ and $\Kt$}\}$, $\S_{\Mt_G}=
\{t \in \S : V(t) $ intersects $\Lt$, $\Mt_G$ and $P_{n,g}\}$,
$\S_{\Mb_G}= \{t \in \S : V(t) \mbox{ intersects $\Lb$, $\Mb_G$ and
  $\Lt$}\}$, $\S_D = \{t \in \S : V(t) \subseteq V(D)\}$, $\S_G = \{t
\in \S : V(t) \subseteq V(G)\}$, and $\S_{GD} = \{t \in \S : V(t)
\mbox{ intersects}$ $V(G)$ and $V(D)\}$.  Notice that $\S_G, \S_{GD},
\S_D$ is a partition of $\S$.

\begin{claim}
\label{cl:D2goodtriangles}
If there exists a perfect packing $\S$ of $\T$, then $|\S_{\Mt_D}| =
m$ and $|\S_{\Mb_D}| = (g-1)m$. This implies that $V(\S_{\Mt_D} \cup
S_{\Mb_D}) \cap V(\Kt) = V(\Kt)$, meaning that the vertices of $\Kt$
are entirely used by $\S_{\Mt_D} \cup S_{\Mb_D}$.
\end{claim}
\begin{proof}
We have $|\S_{\Mt_D}| \leq m$ since $|\Mt_D|= m$. We obtain the
equality since the vertices of $\Mt_D$ only lie in the span of
backward arcs from $P'_{m,g}$ to $\Kt$, and they are not the head or
the tail of a backward arc in $\T$. Thus, the only way to use vertices
of $\Mt_D$ is to create triangles in $\S_{\Mt_D}$, implying
$|\S_{\Mt_D}| \ge m$.  Using the same kind of arguments we also get
$|\S_{\Mb_D}| = (g-1)m$.  As $|V(\Kt)|=gm$ we get the last part of the
claim.
\end{proof}

\begin{claim}
\label{cl:tildeKgoodtriangle}
If there exists a perfect packing $\S$ of $\T$, then there exists $q_0
\in [g]$ such that $\Kt_\S=\Kt_{q_0}$, where $\Kt_\S=\Kt \cap
V(\S_{\Mt_D})$.
\end{claim}
\begin{proof}
Let $\S_{P'}$ be the triangles of $\S$ with at least one vertex in
$P'_{m,g}$.  As according to Claim~\ref{cl:D2goodtriangles} vertices
of $\Kt$ are entirely used by $\S_{\Mt_D} \cup S_{\Mb_D}$, the only
way to consume vertices of $P'_{m,g}$ is by creating local triangles
in $P'_{m,g}$ or triangles in $\S_{\Mt_D}$. In particular, we cannot
have a triangle $(u,v,w)$ with $\{u,v\} \subseteq \Kt$ and $w \in
P'_{m,g}$, or with $u \in \Kt$ and $\{v,w\} \subseteq P'_{m,g}$.  More
formally, we get the partition $\S_{P'}=\S_{\subseteq P'} \cup
\S_{\Mt_D}$.  As $\S$ is a perfect packing and uses in particular all
vertices of $P'_{m,g}$ we get $|V(\S_{P'})|=|V(P'_{m,g})|$, implying
$|V(\S_{\subseteq P'})|=|V(P'_{m,g})|-m$ by
Claim~\ref{cl:D2goodtriangles}.  By Lemma~\ref{lem:path}, this implies
that there exists $q_0 \in [g]$ such that $X' \subseteq X^{'q_0}$
where $X'=V(P'_{m,g}) \setminus V(\S_{\subseteq P'})$.  As $X'$ are
the only remaining vertices that can be used by triangles of
$\S_{\Mt_D}$, we get that the $m$ triangles of $\S_{\Mt_D}$ are of the
form $(u,v,w)$ with $u \in \Kt_{q_0}$, $v \in \Mt_D$, and $w \in X'$.

\end{proof}

\begin{claim}
\label{cl:Dgoodremainingtriangles}
If there exists a perfect packing $\S$ of $\T$, then there exists $q_0
\in [g]$ such that $V(\S_{P'} \cup \S_{\Mt_D} \cup \S_{\Mb_D}) = V(D)
\setminus \Kb_{q_0}$.

\end{claim}
\begin{proof}
By Claim~\ref{cl:D2goodtriangles} we know that $|\S_{\Mb_D}| =
(g-1)m$. As by Claim~\ref{cl:tildeKgoodtriangle} there exists $q_0 \in
[g]$ such that $\Kt_\S=\Kt_{q_0}$, we get that the $(g-1)m$ triangles
of $\S_{\Mb_D}$ are of the form $(u,v,w)$ with $u \in \Kb \setminus
\Kb_{q_0}$, $v \in \Mb_D$, and $w \in \Kt \setminus \Kt_{q_0}$.
\end{proof}

\begin{lemma}
\label{cl:gd}
If there exists a perfect packing $\S$ of $\T$, then $V(\S_{GD}) \cap
V(G) \subseteq V(\Lb)$.  Informally, triangles of $\S_{GD}$ do not use
any vertex of $\Mb_G, \Lt, \Mt_T$ and $P_{n,g}$.
\end{lemma}

\begin{proof}
By Claim~\ref{cl:Dgoodremainingtriangles}, there exists $q_0 \in [g]$
such that $V(\S_{P'} \cup \S_{\Mt_D} \cup \S_{\Mb_D}) = V(D) \setminus
\Kb_{q_0}$.  By Theorem~\ref{thm:nphperfectv2} we know that $\Kb_{q_0}
= K_{(q_0,1)}\dots K_{(q_0,m')}$ for some $m'$ (we even have $m' =
\frac{m}{2}$), where for each $j \in [m']$ we have $V(K_{(q_0,j)}) =
(\theta_j,c_j)$.  Moreover, for any $p \in [g]$, the last property of
Theorem~\ref{thm:nphperfectv2} ensures that for any $a \in
\bA{A}(\T_{(p,q_0)})$, $V(a) \cap V(\Kb_{q_0}) \neq \emptyset$ implies
$a=vc_j$ for $v \in \Lb_p$.  So no arc of $\bA{A}(\T_{(p,q_0)})$, and
thus no arc of $\T$ is entirely included in $\Kb_{q_0}$. This implies
that $\S$ cannot cover the vertices of $\Kb_{q_0}$ using triangles $t$
with $V(t) \subseteq V(\Kb_{q_0})$, and thus that all these vertices
must be used by triangles of $\S_{GD}$, implying that $V(\S_{GD}) \cap
V(D) = \Kb_{q_0}$. The last property of Theorem~\ref{thm:nphperfectv2}
also implies that all the $\theta_j$ have a left degree equal to $0$
in $\T$, or equivalently that there is no arc $a$ of $\T$ such that
$t(a)=\theta_j$ and $h(a) < \theta_j$.  Thus, by induction for any $j$
from $m'$ to $1$, we can prove that the only way for triangles of
$\S_{GD}$ to use $\theta_j$ is to create a triangle
$t_j=(v,\theta_j,c_j)$ with necessarily $v \in V(\Lb)$.
\end{proof}

Lemma~\ref{cl:gd} will allow us to prove
Claims~\ref{cl:G2goodtriangles}, ~\ref{cl:tildeLgoodtriangle}
and~\ref{cl:Ggoodremainingtriangles} using the same arguments as in
the right part ($D$) of the tournament as all vertices of $\Mb_G, \Lt,
\Mt_T$ and $P_{n,g}$ must be used by triangles in $\S_G$.

\begin{claim}
\label{cl:G2goodtriangles}
If there exists a perfect packing $\S$ of $\T$, then $|\S_{\Mt_G}| =
n$ and $|\S_{\Mb_G}| = (g-1)n$. This implies that $V(\S_{\Mt_G} \cup
S_{\Mb_G}) \cap V(\Lt) = V(\Lt)$, meaning that vertices of $\Lt$ are
entirely used by $\S_{\Mt_G} \cup S_{\Mb_G}$.
\end{claim}

\begin{proof}
We have $|\S_{\Mt_G}| \leq n$ since $|\Mt_G|= n$. Lemma~\ref{cl:gd} implies that all vertices of $\Mt_G$ must be used by triangles of $\S_G$, and thus using arcs whose both endpoints lie in $V(G)$.
As vertices of $\Mt_G$ are not the head or the tail of a backward arc in $\T$, we get that the only way for $\S_G$ to use vertices of $\Mt_G$ is to create triangles in $\S_{\Mt_G}$, implying $|\S_{\Mt_G}| \ge n$.
Using the same kind of arguments (and as all vertices of $\Mb_G$ must also be used by triangles of $\S_G$) we also get $|\S_{\Mb_G}| = (g-1)n$.
As $|V(\Lt)|=gn$ we get the last part of the claim.
\end{proof}

\begin{claim}
\label{cl:tildeLgoodtriangle}
If there exists a perfect packing $\S$ of $\T$, then there exists $p_0
\in [g]$ such that $\Lt_\S=\Lt_{p_0}$, where $\Lt_\S=\Lt \cap
V(\S_{\Mt_G})$.
\end{claim}

\begin{proof}
Lemma~\ref{cl:gd} implies that all vertices of $\Mt_G$ and $P_{(n,g)}$ must be used by triangles in $\S_G$.
Let $\S_{P}$ be the triangles of $\S_G$ with at least one vertex in $P_{n,g}$.
As according to Claim~\ref{cl:G2goodtriangles} vertices of $\Lt$ are entirely used by $\S_{\Mt_G} \cup S_{\Mb_G}$, the only way for $\S_G$ to consume vertices of $P_{n,g}$ 
is by creating local triangles in $P_{n,g}$ or triangles in $\S_{\Mt_G}$. In particular, we cannot have a triangle $(u,v,w)$ with $\{u,v\} \subseteq \Lt$ and $w \in P_{n,g}$, or with $u \in \Lt$ and $\{v,w\} \subseteq P_{n,g}$. More formally, we get the partition $\S_{P}=\S_{\subseteq P} \cup \S_{\Mt_G}$.
As $\S_G$ uses in particular all vertices of $P_{n,g}$ we get $|V(\S_{P})|=|V(P_{n,g})|$, implying $|V(\S_{\subseteq P})|=|V(P_{n,g})|-n$ by Claim~\ref{cl:G2goodtriangles}.
By Lemma~\ref{lem:path}, this implies that there exists $p_0 \in [g]$ such that $X \subseteq X^{p_0}$ where $X=V(P_{n,g}) \setminus V(\S_{\subseteq P})$.
As $X$ are the only remaining vertices that can be used by triangles of $\S_{\Mt_G}$, we get that the $n$ triangles of $\S_{\Mt_G}$ are of the form $(u,v,w)$ with $u \in \Lt_{p_0}$, $v \in \Mt_G$, and $w \in X$.
\end{proof}

\begin{claim}
\label{cl:Ggoodremainingtriangles}
If there exists a perfect packing $\S$ of $\T$, then there exists $p_0
\in [g]$ such that $V(\S_{P} \cup \S_{\Mt_G} \cup \S_{\Mb_G}) = V(G)
\setminus \Lb_{p_0}$.
\end{claim}

\begin{proof}
By Claim~\ref{cl:D2goodtriangles} we know that $|\S_{\Mb_G}| = (g-1)n$. As by Claim~\ref{cl:tildeLgoodtriangle} there exists $p_0 \in [g]$ such that $\Lt_\S=\Lt_{p_0}$, 
we get that the $(g-1)n$ triangles of $\S_{\Mb_G}$ are of the form $(u,v,w)$ with $u \in \Lb \setminus \Lb_{p_0}$, $v \in \Mb_G$, and $w \in \Lt \setminus \Lt_{p_0}$.
\end{proof}

We are now ready to state our final claim is now straightforward as
according Claim~\ref{cl:Dgoodremainingtriangles}
and~\ref{cl:Ggoodremainingtriangles} we can define $\S_{(p_0,q_0)}=\S
\setminus ((\S_{P'} \cup \S_{\Mt_D} \cup \S_{\Mb_D}) \cup (\S_{P} \cup
\S_{\Mt_G} \cup \S_{\Mb_G}))$.
\begin{claim}
\label{cl:mainclaim}
If there exists a perfect packing $\S$ of $\T$, there exists $p_0, q_0
\in [g]$ and $\S_{(p_0,q_0)} \subseteq \S$ such that
$V(\S_{(p_0,q_0)}) = V(\T_{(p_0,q_0)})$ (or equivalently such that
$\S_{(p_0,q_0)}$ is a perfect packing of $\T_{(p_0,q_0)}$).
\end{claim}

\paragraph*{Proof of the weak composition}

\begin{theorem}
For any $\epsilon>0$, \perfectpb~(parameterized by the total number of
vertices $N$) does not admit a polynomial (generalized) kernelization
with size bound $\O(N^{2-\epsilon})$ unless ${\sf NP} \subseteq {\sf coNP / Poly}$.
\end{theorem}
\begin{proof}
Given $t$ instances $\{\Fb_l\}$ of \perfectpb~restricted to instances
of Theorem~\ref{thm:nphperfectv2}, we define an instance $\T$ of
\perfectpb~ as defined in Section~\ref{sec:kernel}.  We recall that
$g = \sqrt{t}$, and that for any $l \in [t]$, $|V(L_l)|=n$ and
$|V(K_l)|=m$.  Let $N = |V(\T)|$. As
$N=|V(P'_{(m,g)})|+m+(g-1)m+2mg+|V(P_{(n,g)})|+n+(g-1)n+2ng$ and
$|V(P_{(\m,\g)})| = O(\m\g)$ by Lemma~\ref{lem:path}, we get $N =
\O(g(n+m))=\O(t^{\frac{1}{2+o(1)}} \max(|\Fb_l|))$.  Let us now verify
that there exists $l \in [t]$ such that $\Fb_l$ admits a perfect
packing iff $\T$ admits a perfect packing.
First assume that there exist $p_0,q_0 \in [g]$ such that $\Fb_{(p_0,q_0)}$
admits a perfect packing.  By Lemma~\ref{cl:mainclaim}, there is a
packing $\S_{P'}$ of $P'_{(m,g)}$ such that $V(\S_{p'})= V(P'_{(m,g)})
\setminus X^{'q_0}$.  We define a set $\S_{\Mt_D}$ of $m$ vertex
disjoint triangles of the form $(u,v,w)$ with $u \in \Lt_{q_0}, v \in
\Mt_D, w \in X^{'q_0}$.  Then, we define a set $\S_{\Mb_D}$ of
$(g-1)m$ vertex disjoint triangles of the form $(u,v,w)$ with $u \in
\Lb \setminus \Lb_{q_0}, v \in \Mb_D, w \in \Lt \setminus \Lt_{q_0}$.
In the same way we define $\S_{P}$, $\S_{\Mt_G}$ and
$\S_{\Mb_G}$. Observe that $V(\T) \setminus ((\S_{P'} \cup \S_{\Mt_D}
\cup \S_{\Mb_D}) \cup (\S_{P} \cup \S_{\Mt_G} \cup
\S_{\Mb_G}))=\Kb_{q_0} \cup \Lb_{p_0}$, and thus we can complete our
packing into a perfect packing of $\T$ as $\Fb_{(p_0,q_0)}$ admits a
perfect packing.
Conversely if there exists a perfect packing $\S$ of $\T$, then by
Claim~\ref{cl:mainclaim} there exists $p_0, q_0 \in [g]$ and
$\S_{(p_0,q_0)} \subseteq \S$ such that $V(\S_{(p_0,q_0)}) =
V(\T_{(p_0,q_0)})$, implying that $\Fb_{(p_0,q_0)}$ admits a perfect
packing.
\end{proof}

\begin{corollary}
\label{theo:noKernel}
For any $\epsilon>0$, \pb~(parameterized by the size $k$ of the solution)
does not admit a polynomial kernel with size
 $\O(k^{2-\epsilon})$ unless ${\sf NP} \subseteq {\sf coNP / Poly}$.

\end{corollary}

\section{Conclusion and open questions}

Concerning approximation algorithms for \pb~restricted to sparse
instances, we have provided a $(1+\frac{6}{c+5})$-approximation
algorithm where $c$ is a lower bound of the ${\mathtt minspan}$ of the
instance. On the other hand, it is not hard to solve by dynamic
programming \pb~for instances where ${\mathtt maxspan}$ is bounded
above. Using these two opposite approaches it could be interesting to
derive an approximation algorithm for \pb~with factor better than
$4/3$ even for sparse tournaments.

Concerning {\sf FPT} algorithms, the approach we used for sparse
tournament (reducing to the case where $m=\O(k)$ and apply the $\O(m)$
vertices kernel) cannot work the general case. Indeed, if we were able
to sparsify the initial input such that $m'=\O(k^{2-\epsilon})$,
applying the kernel in $\O(m')$ would lead to a tournament of total
bit size (by encoding the two endpoint of each arc)
$\O(m'log(m'))=\O(k^{2-\epsilon})$, contradicting
Corollary~\ref{theo:noKernel}.  Thus the situation for \pb~could be as
in vertex cover where there exists a kernel in $\O(k)$ vertices,
derived from~\cite{linearVc74}, but the resulting instance cannot have
$\O(k^{2-\epsilon})$ edges~\cite{dell2014kernelization}. So it is
challenging question to provide a kernel in $\O(k)$ vertices for the
general \pb~problem.

\bibliography{TP-biblio}

%\pagebreak

\appendix

\section{Definitions}
\label{appendix-def}

\paragraph*{Approximation}
\begin{definition}[\cite{papadimitriou1991optimization}]\label{def:L}
Let $\Pi$ and $\Pi'$ be two optimization (maximization or minimization) problems.
We say that $\Pi$ $L$-reduces to $\Pi'$ if there are two polynomial-time algorithms $f$, $g$, and
constants $\alpha, \beta > 0$ such that for each instance $I$ of $\Pi$
\begin{itemize}
\item[(a)] Algorithm $f$ produces an instance $I' = f(I)$ of $\Pi'$ such that the optima of
$I$ and $I'$, $OPT(I)$ and $OPT(I')$, respectively, satisfy $OPT(I') \le \alpha OPT(I)$
\item[(b)] Given any solution of $I'$ with cost $c$, algorithm $g$ produces a solution of
$I$ with cost $c$ such that $|c - OPT(I)| \le \beta |c'- OPT(I')|$.
\end{itemize}
\end{definition}
\begin{definition}
Let $A$ be an algorithm of a maximization (resp. minimization) problem $\Pi$. For $\rho \geq 1$, we say that $A$ is a $\rho$-approximation of $\Pi$ iff for any instance $I$ of $\Pi$, $A_I \geq OPT(I)/\rho$ (resp. $A_I \leq \rho OPT(I)$) where $A_I$ is the value of the solution $A(I)$ and $OPT(I)$ the value of a optimal solution of $I$.
\end{definition}

\begin{definition}
\label{def:apx}
Let $\Pi$ be a {\sf NP}-optimization problem. The problem $\Pi$ is in {\sf APX} if there exists a constant $\rho >1$ such that $\Pi$ admits a $\rho$-approximation algorithm.
\end{definition}

\begin{definition}
\label{def:ptas}
Let $\Pi$ be a {\sf NP}-optimization problem. The problem $\Pi$ admits a {\sf PTAS} if for any $\epsilon >0$, there exists a polynomial $(1 + \epsilon)$-approximation of $\Pi$.
\end{definition}

\paragraph*{Parameterized complexity} We refer the reader 
to~\cite{downey2013fundamentals} for more details on parameterized
complexity and kernelization, and we recall here only some basic
definitions. A \emph{parameterized problem} is a language $L \subseteq
\Sigma^* \times \mathbb{N}$.  For an instance $I=(x,k) \in \Sigma^*
\times \mathbb{N}$, the integer $k$ is called the \emph{parameter}.

A parameterized problem is \emph{fixed-parameter tractable} ({\sf
  FPT}) if there exists an algorithm $A$, a computable function $f$,
and a constant $c$ such that given an instance $I=(x,k)$, $A$ (called
an {\sf FPT} algorithm) correctly decides whether $I \in L$ in time
bounded by $f(k) \cdot |I|^c$, where $|I|$ denotes the size of $I$.
Given a computable function $g$, a \emph{kernelization algorithm} (or
simply a \emph{kernel}) for a parameterized problem $L$ of \emph{size}
$g$ is an algorithm $A$ that given any instance $I=(x,k)$ of $L$, runs
in polynomial time and returns an equivalent instance $I'=(x',k')$
with $|I'|+k' \le g(k)$.  It is well-known that the existence of an
{\sf FPT} algorithm is equivalent to the existence of a kernel (whose
size may be exponential), implying that problems admitting a
polynomial kernel form a natural subclass of {\sf FPT}. Among the wide
literature on polynomial kernelization, we only recall in
the notion of weak composition used to
lower bound the size of a kernel.

\begin{definition}[Definition as written in~\cite{jansen2015sparsification}]\label{def:orcompo}
Let $L \subseteq \Sigma^*$ be a language, $R$ be a polynomial equivalence relation on $\Sigma^*$, let  $Q \subseteq \Sigma^* \times \mathbb{N}$
be a parameterized  problem, and let $f : \mathbb{N} \rightarrow \mathbb{N}$ be a function. An or-cross-composition of $L$ into $Q$ (with respect to $R$) of cost $f(t)$ is an
algorithm that, given $t$ instances $x_i \in \Sigma^*$ of $L$ belonging to the same equivalence class of $R$, takes time polynomial in $\sum_{i \in [t]}|x_i|$ and outputs an instance 
$(y, k) \in \Sigma^* \times \mathbb{N}$ such that:
 \begin{enumerate}
  \item the parameter $k$ is bounded by $\O(f(t)\max_i|x_i|^c)$, where $c$ is some constant independent of $t$, and
  \item $(y, k) \in Q$ if and only if there is an $i \in [t]$ such that $x_i \in L$.
 \end{enumerate}
\end{definition}

\begin{theorem}[\cite{bodlaender2014kernelization}]\label{thm:orcompo}
Let $L \subseteq \Sigma^*$ be a language, let $Q \subseteq \Sigma^* \times \mathbb{N}$ be a parameterized problem, and
let $d, \epsilon$ be positive reals. If $L$ is {\sf NP}-hard under Karp reductions, has an or-cross-composition
into $Q$ with cost $f(t) = t^{1/d+o(1)}$ , where $t$ denotes the number of instances, and $Q$ has a
polynomial (generalized) kernelization with size bound $\O(k^{d-\epsilon})$, then ${\sf NP} \subseteq {\sf coNP / Poly}$.
\end{theorem}

\section{Problems}
\label{app:prbl}
\begin{problem}\label{def:fvs}\hspace{-0.3em}(FVS) ~\\
\textbf{Input:} A directed graph $D=(V,A)$.~\\
\textbf{Output:} A set of vertices $X\subseteq V$ such that $D[V\setminus X]$ is acyclic.~\\
\textbf{Optimisation:} Minimise $|X|$.~\\
The problem is called FVST if the input is a tournament.
\end{problem}

\begin{problem}\label{def:dsp}\hspace{-0.3em}($d$-{\sc Set Packing}) ~\\
\textbf{Input:} An integer $d \geq 3$ and a $d$-uniform hypergraph \mbox{$G=(V,H)$}.~\\
\textbf{Output:} A subset of hyperedges $X = \{X_i, i\in [k] $ with $X_i \in H\}$ such that for every $i \neq j$, $X_i \cap X_j = \emptyset$.~\\
\textbf{Optimisation:} Maximise $k$.
\end{problem}

\begin{problem}\label{def:pdsp}\hspace{-0.3em}({\sc Perfect} $d$-{\sc Set Packing}) ~\\
\textbf{Input:} An integer $d \geq 3$ and a $d$-uniform hypergraph \mbox{$G=(V,H)$}.~\\
\textbf{Question:} Is there a subset of hyperedges $X = \{X_i, i\in [k] $ with $X_i \in H\}$ such that for every $i \neq j$, $X_i \cap X_j = \emptyset$ and $\bigcup_{i\in[k]}X_i = V$?
\end{problem}

\begin{problem}\label{def:hp}\hspace{-0.3em}($H$-{\sc Packing})~\\
\textbf{Input:} A graph \mbox{$G=(V,E)$} and a subgraph $H$.~\\
\textbf{Output:} A collection of subgraphs $X=\{H_i, i\in [k]\}$ such that for every $i$, $H_i$ is isomorphic to $H$ and for every $j \neq i$, $V(H_i) \cap V(H_j) = \emptyset$.~\\
\textbf{Optimisation:} Maximise $k$.
\end{problem}

\begin{problem}\label{def:php}\hspace{-0.3em}({\sc Perfect} $H$-{\sc Packing})~\\
\textbf{Input:} A graph \mbox{$G=(V,E)$} and a subgraph $H$.~\\
\textbf{Question:} Is there a collection of subgraphs $X=\{H_i, i\in [k]\}$ such that for every $i$, $H_i$ is isomorphic to $H$, for every $j \neq i$, $V(H_i) \cap V(H_j) = \emptyset$ and $\bigcup_{i\in[k]}H_i = V$?
\end{problem}

\section{Polynomial detection of sparse tournaments}\label{app:proof}
%\subsection{Proof of Section~\ref{sec:notation}}
\begin{lemma}
\label{lem:faslinear}
In polynomial time, we can decide if a tournament is sparse or not, and if so, to give a linear representation whose FAS is a matching
\end{lemma}
\begin{proof}
Indeed if a tournament $\T$ is sparse we can detect the
first vertex (or vertices) of a linear representation $\V(\T)$ of $\T$
where $\bA{A}(\T)$ is a matching. If $T$ has a vertex $x$ of indegree
0 then $x$ must be the first or the
second vertex of $\V(\T)$, and
we can always suppose that x is the first vertex of $\V(\T)$. Otherwise, we look at
$Z$ the set of vertices of $\T$ with indegree 1. As $\T$ is a
tournament we have $|Z|\le 3$ and if $Z=\emptyset$ then $T$ is not a
sparse tournament. If $|Z|=1$, then the only element of $Z$ must be
the first vertex of $\V(\T)$. If $|Z|=2$ with $Z=\{x,y\}$ such that
$xy$ is an arc of $\T$, then $x$ must be the first element of $\V(\T)$
and $y$ its second element. Finally, if $|Z|=3$ with $Z=\{x,y,z\}$ then
$xyz$ must be a triangle of $\T$ and must be placed at the beginning
of $\V(\T)$. So repeating inductively these arguments we obtain in
polynomial time in $|\T|$ either $\V(\T)$ such that $\bA{A}(\T)$ is a
matching or a certificate that $\T$ is not sparse. 
\end{proof}

\end{document}